\documentclass[a4paper,11pt]{article}

\usepackage{amsmath,amsthm}
\usepackage{amssymb}
\usepackage{mathtools}
\usepackage{subcaption}
\usepackage[ruled,linesnumbered]{algorithm2e}
\usepackage[disable]{todonotes}
\usepackage{comment}
\usepackage{witharrows}
\usepackage{complexity}
\usepackage{accsupp}
\usepackage{subcaption}
\usepackage{xpatch}
\usepackage{apptools}
\usepackage{hyperref}
\usepackage[sort&compress,nameinlink,noabbrev,capitalize]{cleveref}
\usepackage{paralist}
\usepackage{thm-restate}
\usepackage{fullpage}

\usepackage{cellspace}
\newcommand\cincludegraphics[2][]{\raisebox{-0.3\height}{\includegraphics[#1]{#2}}}

\usepackage{booktabs}

\newcommand{\pPVSlong}{\textsc{$\Pi$ Vertex Splitting}}
\newcommand{\pPVS}{\textsc{$\Pi$\nobreakdash-VS}}
\newcommand{\pPVSa}[1]{\textsc{#1\nobreakdash-VS}}
\newcommand{\pFVS}[1]{\textsc{\ensuremath{\free_\prec(#1)}\nobreakdash-VS}}

\newcommand{\F}{\ensuremath{\mathcal{F}}}

\usepackage{tabularx, environ}
\makeatletter
\newcolumntype{\expand}{}
\long\@namedef{NC@rewrite@\string\expand}{\expandafter\NC@find}
\NewEnviron{problem}[2][]{%
  \def\problem@arg{#1}%
  \def\problem@framed{framed}%
  \def\problem@lined{lined}%
  \def\problem@doublelined{doublelined}%
  \ifx\problem@arg\@empty%
    \def\problem@hline{}%
  \else%
    \ifx\problem@arg\problem@doublelined%
      \def\problem@hline{\hline\hline}%
    \else%
      \def\problem@hline{\hline}%
    \fi%
  \fi%
  \ifx\problem@arg\problem@framed%
    \def\problem@tablelayout{|>{\bfseries}lX|c}%
    \def\problem@title{\multicolumn{2}{|l|}{%
        \raisebox{-\fboxsep}{\textsc{\large #2}}%
      }}%
  \else
    \def\problem@tablelayout{>{\bfseries}lXc}%
    \def\problem@title{\multicolumn{2}{l}{%
        \raisebox{-\fboxsep}{\textsc{\large #2}}%
      }}%
  \fi%
  \medskip
  \par\noindent%
  \begin{tabularx}{\textwidth}{\expand\problem@tablelayout}%
    \problem@hline%
    \problem@title\\[2\fboxsep]%
    \BODY\\\problem@hline%
  \end{tabularx}%
  \medskip
  \par%
}
\makeatother

\def\barroman#1{\sbox0{#1}\dimen0=\dimexpr\wd0+1pt\relax
  \makebox[\dimen0]{\rlap{\vrule width\dimen0 height 0.06ex depth 0.06ex}%
    \rlap{\vrule width\dimen0 height\dimexpr\ht0+0.03ex\relax 
            depth\dimexpr-\ht0+0.09ex\relax}%
    \kern.5pt#1\kern.5pt}}

\newcommand{\cupdot}{\mathbin{\mathaccent\cdot\cup}}
\DeclareMathOperator*{\argmin}{\arg\!\min}
\DeclareMathOperator{\free}{Free}
\DeclareMathOperator{\emb}{Emb}
\DeclareMathOperator{\wgt}{wgt}
\DeclareMathOperator{\diam}{diam}
\DeclareMathOperator{\width}{wdt}
\DeclareMathOperator{\Split}{Split}
\DeclareMathOperator{\Constr}{Constr}
\DeclareMathOperator{\AllConstr}{AllConstr}
\DeclareMathOperator{\range}{Range}
\DeclareMathOperator{\dom}{Domain}

\newcommand{\set}[1] {
  \mathchoice
  {\left \{ #1 \right \}}
  {\{ #1 \}}
  {\{ #1 \}}
  {\{ #1 \}}
}

\newcommand{\parens}[1] {
  \mathchoice
  {\left ( #1 \right )}
  {( #1 )}
  {( #1 )}
  {( #1 )}
}

\newcommand{\abs}[1] {
  \mathchoice
  {\left | #1 \right |}
  {| #1 |}
  {| #1 |}
  {| #1 |}
}

\newcommand{\N}[1]{\abs{V(#1)}}
\renewcommand{\M}[1]{\abs{E(#1)}}

\theoremstyle{plain}
\newtheorem{theorem}{Theorem}[section]
\newtheorem{lemma}[theorem]{Lemma}

\newtheorem{proposition}[theorem]{Proposition}

\theoremstyle{definition}
\newtheorem{definition}[theorem]{Definition}

\newif\iflong
\newif\ifshort

\longtrue

\iflong
\else
\shorttrue
\fi

\usepackage{etoolbox}
\newcommand{\appsymb}{$\bigstar$}
\newcommand{\appref}[1]{\hyperref[proof:#1]{\appsymb}}

\newcommand{\appendixsection}[1]{%
  \iflong{}\else{}
    \gappto{\appendixText}{\section{Additional Material for \cref{#1} }\label{app:#1}}
  \fi{}
}

\newcommand{\toappendix}[1]{%
  \iflong{}#1\else{}
    \gappto{\appendixText}
    {
        #1
      }
  \fi{}
}

\newcommand{\appendixproof}[2]{%
  \iflong{}#2\else{}\gappto{\appendixText}
    {
      \subsection{Proof of \cref{#1}}\label{proof:#1}
      #2
    }
  \fi{}
}

\makeatletter
\xpatchcmd{\thmt@restatable}%
{\csname #2\@xa\endcsname\ifx\@nx#1\@nx\else[{#1}]\fi}%
{\IfAppendix{\csname #2\@xa\endcsname}{\csname #2\@xa\endcsname\ifx\@nx#1\@nx\else[{#1}]\fi}}%
{}{} %
\makeatother

\iflong
\makeatletter
\xpatchcmd{\thmt@restatable}%
{\csname #2\@xa\endcsname\ifx\@nx#1\@nx\else[{#1}]\fi}%
{\csname #2\@xa\endcsname}%
{}{} %
\makeatother
\fi

\usepackage{authblk}

\title{On the Complexity of Establishing Hereditary Graph Properties via Vertex Splitting\thanks{MS gratefully acknowledges support by the Alexander von Humboldt Foundation.}} %

\author{Alexander Firbas}%

\author{Manuel Sorge}

\affil{TU Wien, Austria, \texttt{\{alexander.firbas, manuel.sorge\}@tuwien.ac.at}}

\begin{document}

\maketitle

\begin{abstract}
  Vertex splitting is a graph operation that replaces a vertex $v$ with two nonadjacent new vertices and makes each neighbor of $v$ adjacent with one or both of the introduced vertices.
  Vertex splitting has been used in contexts from circuit design to statistical analysis.
  In this work, we explore the computational complexity of achieving a given graph property $\Pi$ by a limited number of vertex splits, formalized as the problem \pPVSlong~(\pPVS).
  We focus on hereditary graph properties and contribute four groups of results:
  First, we classify the classical complexity of \pPVS\ for graph properties characterized by forbidden subgraphs of size at most 3.
  Second, we provide a framework that allows to show \NP-completeness whenever one can construct a combination of a forbidden subgraph and prescribed vertex splits that satisfy certain conditions.
  Leveraging this framework we show \NP-completeness when $\Pi$ is characterized by forbidden subgraphs that are sufficiently well connected.
  In particular, we show that \pPVSa{$F$-Free} is \NP-complete for each biconnected graph $F$.
  Third, we study infinite families of forbidden subgraphs, obtaining \NP-hardness for \pPVSa{Bipartite} and \pPVSa{Perfect}.
  Finally, we touch upon the parameterized complexity of \pPVS\ with respect to the number of allowed splits, showing para-\NP-hardness for \textsc{$K_3$-Free-VS} and deriving an \XP-algorithm when each vertex is only allowed to be split at most once.
\end{abstract}

\section{Introduction}

\emph{Vertex splitting} is the graph operation in which we take a vertex $v$, remove it from the graph, add two \emph{copies} $v_1$, $v_2$ of $v$, and make each former neighbor of $v$ adjacent with $v_1$, $v_2$, or both.
Vertex splitting has been used in circuit design~\cite{PaikRS98,MayerE93}, the visualization of nonplanar graphs in a planar way~\cite{planar_splitting_2001,social_network,eppstein2018planar,planarizing_vs_fpt,baumann2023parameterized}, improving force-based graph layouts~\cite{eades_tension_1995}, in graph clustering with overlaps~\cite{abukhzam_cluster_2018,ArrighiBDSW23,firbas_cluster_2023}, in statistics~\cite{letter_display,davoodi_edge_2016} (see \cite{firbas_cluster_2023}), and variants of vertex splitting play roles in graph theory~\cite{HiltonZ97,mertzios_vertex_2011}, in particular in Fleischner's Splitting Lemma~\cite{fleischner_eulerian_1990} and in Tutte's theorem relating wheels and general three-connected graphs~\cite{tutte_connnectivity_1966}.
Vertex splitting can also be thought of as an inverse operation of vertex contraction, which is the underlying operation of the twinwidth graph parameter~(e.g.~\cite{BonnetKTW22}).

In some of the above applications, we are given a graph and want to establish a graph property by splitting the least number of times:
In circuit design, we aim to bound the longest path length~\cite{PaikRS98,MayerE93}, when visualizing non-planar graphs we aim to establish planarity~\cite{eppstein2018planar,planar_splitting_2001,planarizing_vs_fpt} or pathwidth one~\cite{baumann2023parameterized}, and in statistics and when clustering with overlaps we want to obtain a cluster graph (a disjoint union of cliques)~\cite{letter_display,davoodi_edge_2016,abukhzam_cluster_2018,ArrighiBDSW23,firbas_cluster_2023}.
This motivates generalizing these problems by letting $\Pi$ be any graph property (a family of graphs) and studying the problem \pPVSlong~(\pPVS): Given a graph $G$ and an integer $k$, is it possible to apply at most~$k$ vertex split operations to $G$ to obtain a graph in $\Pi$?
The above-mentioned graph properties are closed under taking induced subgraphs and thus we mainly focus on this case.
For graph operations different from vertex splitting the complexity of establishing graph properties~$\Pi$ is well studied, such as for deleting vertices~(e.g.~\cite{node_deletion,node_deletion_fpt}), adding or deleting edges~(see the recent survey~\cite{CrespelleDFG23}), or edge contractions (e.g.~\cite{GolovachHP13,GuillemotM13,HeggernesHLP13,GuoC15}).
In this work, we aim to start this direction for vertex splitting, that is, how can we characterize for which graph properties \pPVS\ is tractable?
Our main focus here is the classical complexity, that is, \NP-hardness vs.\ polynomial-time solvability, but we also provide modest contributions to the parameterized complexity with respect to the number of allowed splits.

\looseness=-1
Our results are as follows.
Each graph property $\Pi$ that is closed under taking induced subgraphs is characterized by a family $\F$ of \emph{forbidden induced subgraphs}.
We also write $\Pi$ as $\free_\prec(\F)$.
It is thus natural to begin by considering small forbidden subgraphs.
We classify for each family \F\ that contains graphs of size at most~3 whether \pFVS{\F} is polynomial-time solvable or \NP-complete (\cref{section:leq3}).
Indeed, it is \NP-complete precisely if $\F$ contains only the path $P_3$ on three vertices or a triangle $K_3$.
All other cases are polynomial-time solvable and this extends also to so-called threshold and split graphs.
In contrast, together with our results below, we obtain \NP-completeness for each connected forbidden subgraph~$F$ with four vertices except for $P_4$s and claws $K_{1, 3}$, for which the complexity remains open.

\looseness=-1
Second, the hardness construction for $K_3$-free graphs indicates that high connectivity in forbidden subgraphs makes \pPVS\ hard and thus we explored this direction further.
We obtain a framework for showing \NP-hardness of \pPVS\ whenever one can use forbidden induced subgraphs to construct certain splitting configurations (\cref{section:biconnected}).
That is, a graph together with a recipe specifying distinguished vertices that will be connected to the outside and how to split them.
Essentially, if one can provide a splitting configuration that avoids introducing new forbidden subgraphs and that decreases the connectivity to the outside well enough, then we can use such a configuration to give a hardness construction.
We then provide such splitting configurations for individual biconnected forbidden induced subgraphs, for families of triconnected forbidden subgraphs of bounded diameter and for families of 4-connected forbidden induced subgraphs, obtaining \NP-hardness of \pFVS{\F} in these cases.

\looseness=-1
Third, the above results do not cover the case where $\F$ is the family of all cycles, and this must be so because \pPVSa{Forest} is polynomial-time solvable~\cite{baumann2023parameterized,firbas_establishing_2023}.
However, we show that, if we forbid only cycles of at most a certain length, or all cycles of odd length, then \pPVS\ becomes \NP-complete again (\cref{section:cycles}).
This hardness extends also to so-called perfect graphs.

\looseness=-1
Finally, we contribute to the parameterized complexity of \pPVS\ with respect to the number~$k$ of allowed vertex splits (\cref{section:parameterized}).
Previously it was known that \pPVS\ is fixed-parameter tractable when $\Pi$ is closed under taking minors~\cite{planarizing_vs_fpt}, when $\Pi = \free_\prec(P_3)$~\cite{firbas_cluster_2023,firbas_establishing_2023}, and when $\Pi$ consists of graphs of pathwidth one or when $\Pi$ is MSO$_2$-definable and of bounded treewidth~\cite{baumann2023parameterized}.
In contrast, we observe that \pFVS{K_3} is \NP-hard even for $k = 2$.
However, the hardness comes from the fact that one can split a vertex multiple times:
If we instead can split each vertex at most once, then we obtain an $O(\sqrt{2}^{k^2} \cdot n^{k + 3})$-time algorithm.

\subsection{Preliminaries}

\subparagraph{General Graph Notation}
Unless explicitly mentioned otherwise, all graphs are undirected and without parallel edges or self-loops.
Given a graph~$G$ with vertex set $V(G)$ and edge set $E(G)$, we denote the neighborhood of a vertex $v\in V(G)$ by $N_G(v)$.
If the graph~$G$ is clear from the context, we omit the subscript~$G$. %
For $V'\subset V(G)$, we write $G[V']$ for the graph induced by the vertices $V'$.
For $u,v\in V(G)$ we write $uv$ as a shorthand for $\{u,v\}$, $G-v$ for $G[V\setminus \{v\}]$, $d_G(v)$ for $|N_G(v)|$, $d_G(u, v)$ for the length of the shortest path from $u$ to $v$, and $\diam(G)$ for the diameter of $G$, that is, $\max_{u, v \in V(G)} d_G(u, v)$.
The complement of a graph $G$ is denoted by $\overline{G}$.
The graph $K_n$ is the complete graph on $n$ vertices and $C_n$ is the cycle graph of $n$ vertices. If a graph $G$ is isomorphic to $H$, we write $G\simeq H$.

We mark all directed graphs $\vec{G}$ with an arrow.
The in-neighborhood is denoted by $ N_{\vec{G}}^-(\cdot)$ and the out-neighborhood by $ N_{\vec{G}}^+(\cdot)$. 
We say the directed graph $\vec{G}$ is an \emph{orientation} of $G$ if the underlying undirected graph of $\vec{G}$ is $G$.

The \emph{$k$-subdivision} of a graph $G$
is defined as the graph that results from
replacing each of $G$'s edges $uv$ with a path $u, p^{uv}_1, p^{uv}_2, \ldots, p^{uv}_k, v$ where $p^{uv}_1, p^{uv}_2, \ldots, p^{uv}_k$ are new vertices.

\subparagraph{Vertex Splitting}
Let $G$ be a graph, $v \in V(G)$, and $V_1, V_2$ subsets of $N_G(v)$ such that $V_1 \cup V_2 = N_G(v)$.
Furthermore, let $v_1$ and $v_2$ denote two fresh vertices, that is, $\set{v_1, v_2} \cap V(G) = \varnothing$.
Consider the graph $G'$ that is obtained from $G$ by deleting $v$, and adding $v_1$ and $v_2$ such that $N_{G'}(v_1) = V_1$ and $N_{G'}(v_2) = V_2$.
Then, we say $G'$ was obtained from $G$ by \emph{splitting} $v$ (via a \emph{vertex split}).
If $V_1 \cap V_2 = \varnothing$, we speak of a \emph{disjoint} vertex split,
and if either $V_1 = \varnothing$ or $V_2 = \varnothing$, we say the split is \emph{trivial}.
Furthermore, we say $v$ was \emph{split into} $v_1$ and $v_2$, and call these vertices the \emph{descendants} of $v$. Conversely, $v$ is called the \emph{ancestor} of $v_1$ and $v_2$.
Finally, consider an edge $v_1w$ (resp. $v_2w$) of $G'$.
We say that the edge $vw$ of $G$ was assigned to $v_1$ (resp. $v_2$) in the split,
and call $v_1w$ (resp. $v_2w$) a \emph{descendant edge} of $vw$. 

A \emph{splitting sequence} of $k$ splits is a sequence of graphs $G_0, G_1, \ldots, G_k$,
such that $G_{i+1}$ is obtainable from $G_{i}$ via a vertex split for $i \in \set{0, \ldots, k - 1}$.
The notion of descendant vertices (resp. ancestor vertices) is extended in a transitive and reflexive way to splitting sequences.

\subparagraph{Embeddings and Hereditary Graph Properties}
For graphs $G$ and $H$,
we write $\emb_\prec(G, H)$ (resp. $\emb_\subseteq(G, H)$) to denote the set of all \emph{induced embeddings} of $G$ in $H$ (resp. \emph{subgraph embeddings}),
that is, the set of all injective $f \colon V(G) \to V(H)$ where
$\forall uv \in V(G)^2 \colon uv \in E(G) \iff f(u)f(v) \in E(H)$ 
(resp. $\forall uv \in V(G)^2 \colon uv \in E(G) \implies f(u)f(v) \in E(H)$ ).
In case $\emb_\prec(G, H) \neq \varnothing$ (resp. $\emb_\subseteq(G, H) \neq \varnothing$), we write $G \prec H$ (resp. $G \subseteq H$) and say $G$ is an \emph{induced subgraph} (resp. a \emph{subgraph}) of $H$.

The \emph{circumference} of a graph $G$ is the biggest $i \in \mathbb{N}$ such that $C_i \subseteq G$.

\looseness=-1
For a set of graphs $\mathcal{F}$, we write
$\free_\prec(\mathcal{F})$ (resp. $\free_\subseteq(\mathcal{F})$)
to denote the set of graphs where $G \in \free_\prec(\mathcal{F})$ (resp. $G \in \free_\subseteq(\mathcal{F})$) if and only if
$\emb_\prec(F, G) = \varnothing$ (resp.  $\emb_\subseteq(F, G) = \varnothing$) for all $F \in \mathcal{F}$.
Set $\mathcal{F}$ is the set of \emph{forbidden induced subgraphs} (resp. \emph{forbidden subgraphs})
that \emph{characterize} the \emph{hereditary (graph) property} $\free_\prec(\mathcal{F})$ (resp. $\free_\subseteq(\mathcal{F})$).

\subparagraph{Miscellaneous}
For a function $f \colon A \to B$, its domain $\dom(f)$ is $A$, and its range, $\range(f)$, is $\set{b \mid \exists a \in A \colon f(a) = b}$.
For a set $X$, we let $\mathcal{P}(X)$ be its power set.

\ifshort
\smallskip
\noindent \textit{Due to space constraints, statements marked with}~\appsymb~\textit{are proved in the Appendix}.%
\fi

\section{Properties Characterized by Small Forbidden Induced Subgraphs: Outline}
\label{section:leq3}
\appendixsection{section:leq3}

\looseness=-1
We now give an outline of the characterization of \pPVSlong~(\pPVS) for families \F\ of forbidden subgraphs with at most 3 vertices, the full version is given in \cref{section:leq3-full}.
First, we can make several trivial observations: If one of $K_0$, $K_1$, $K_2$, or $\overline{K_2}$ is forbidden and it is present in the input graph, then there is no way to destroy these forbidden subgraphs with vertex splitting and hence we can immediately return a failure symbol.
This gives a trivial algorithm if $K_0 \in \F$ or $K_1 \in \F$.
Moreover, if $\overline{K_2} \in \F$, then the input graph is a clique or we can return failure.
Since splitting introduces a $\overline{K_2}$, instance $(G, k)$ is positive if and only if $(G, 0)$ is positive, which we can check in polynomial time.
Similarly, if $K_2 \in \F$, then the input graph is an independent set or we can return failure.
Through splitting, we can only introduce more independent vertices and thus $(G, k)$ is positive if and only if $(G, 0)$ is positive.

It follows that we can focus on families \F\ that contain subgraphs with exactly 3 vertices, that is, $\F \subseteq \{P_3, \overline{P_3}, K_3, \overline{K_3}\}$.
If \F\ contains $\overline{P_3}$ or $\overline{K_3}$ but neither $P_3$ nor $K_3$, then we have a similar observation as above: $\overline{P_3}$ and $\overline{K_3}$ cannot be destroyed by vertex splits and thus $(G, k)$ is positive if and only if $(G, 0)$ is, which is checkable in polynomial time.

It thus remains to classify families $\F \subseteq \{P_3, \overline{P_3}, K_3, \overline{K_3}\}$ that contain $P_3$ or $K_3$.
If $\F = \{P_3\}$ then \pFVS{\F}\ is \NP-complete by a result of Firbas et al.~\cite[Theorem 4.4]{firbas_cluster_2023}.
If $\F = \{K_3\}$ then \NP-completeness follows from \cref{theorem:single_biconnected_np_complete} or \cref{theorem:paraNPhard}, which we prove below.
However, if we add $\overline{P_3}$ and/or $\overline{K_3}$ then the problems turn polynomial-time again for subtle and different reasons:

\looseness=-1
In the case where $\{K_3, \overline{K_3}\} \subseteq \F$ we can apply Ramsey-type arguments to show that an algorithm only needs to check for a constant number of different yes-instances.
If $\overline{P_3} \in \F$ we can observe that destroying any $P_3$ or $K_3$ necessarily introduces a $\overline{P_3}$, which cannot be removed afterwards.
This takes care of all cases for \F\ (see \cref{table:leq3} in \cref{section:leq3-full}) except $\F = \{P_3, \overline{K_3}\}$.
For this case we can observe that the graphs resulting from a splitting solution are \emph{cluster graphs}, disjoint unions of cliques, with at most two clusters (cliques).
As $\overline{K_3}$ cannot be destroyed by vertex splitting, the input graph may only contain $P_3$s.
Furthermore, $P_3$s can only be destroyed by splitting their midpoints.
It is thus intuitive that the input graph of a yes-instance must consist of two cliques that may overlap and, furthermore, the overlap must not exceed the number~$k$ of allowed splits.
This is indeed what we can show.
We use the following characterization of $P_3$-free vertex splittings in terms of so-called sigma clique covers by Firbas et al.~\cite{firbas_cluster_2023}:
\begin{definition}\label{definition:sccshortversion}
  Let $G$ be a graph. Then, $\mathcal{C} \subseteq \mathcal{P}(V)$ is called a \emph{sigma clique cover} of 
  $G$ if
  \iflong    
    \begin{enumerate}
    \item $G[C]$ is a clique for all $C \in \mathcal{C}$ and
    \item for each $e \in E(G)$, there is $C \in \mathcal{C}$ such that $e \in E(G[C])$, that is, all edges of $G$ are ``covered'' by some clique of $\mathcal{C}$.
    \end{enumerate}
    The \emph{weight} of a sigma clique cover $\mathcal{C}$ is denoted by $\wgt(\mathcal{C})$, where
    \[
      \wgt(\mathcal{C}) \coloneqq \sum_{C \in \mathcal{C}} |C|.    
    \]  

  \else
    \begin{inparaenum}
    \item $G[C]$ is a clique for all $C \in \mathcal{C}$ and
    \item for each $e \in E(G)$, there is $C \in \mathcal{C}$ such that $e \in E(G[C])$, that is, all edges of $G$ are ``covered'' by some clique of $\mathcal{C}$.
    \end{inparaenum}
    The \emph{weight} of a sigma clique cover $\mathcal{C}$ is $\wgt(\mathcal{C}) \coloneqq \sum_{C \in \mathcal{C}} |C|$.
  \fi      
\end{definition}

\begin{lemma}[Firbas et al.~\cite{firbas_cluster_2023}, Lemma 4.3]
  Let $G = (V,E)$ be a graph, and let $I \coloneqq \{v \in V \mid d_G(v) = 0\}$.
  Then, there are at most $k$ vertex splits that turn $G$ into a cluster graph if and only if $G$ admits a sigma clique cover with weight at most $|V| - |I| + k$.
\end{lemma}
Intuitively, the sets of the sigma clique cover correspond exactly to the clusters of the cluster graph obtained after splitting.
It now follows that, if our input graph indeed consists of two cliques that overlap in at most $k$ vertices, then there is a solution to \pFVS{\{P_3, \overline{K_3}\}}.
The more interesting direction is the reverse one.
That is, all yes-instances indeed look as such.
This is essentially proved in the following lemma.
\begin{lemma}
    Let $G = (V, E) \in \free_\prec(\set{\overline{K_3}})$ without isolated vertices
    and let $M \coloneqq \set{ v \in V \mid \exists f \in \emb_\prec(P_3, G) \colon v \in f(V(P_3)) \land d_{P_3}(f^{-1}(v)) = 2}$, that is, the set of all vertices in $G$ that are a midpoint of some induced $P_3$ in $G$. If $G[M]$ is a non-empty complete graph, then there are $C_1, C_2 \subseteq V$ such that $\set{C_1, C_2}$ is a sigma clique cover of $G$ with $C_1 \cap C_2 = M$ and $\wgt(\set{C_1, C_2}) = |V| + |M|$.
\end{lemma}
\begin{proof}
    First, we get some trivial cases out of the way.
    The graph $G$ cannot have more than two connected components, for then we would have $\overline{K_3} \prec G$. If $G$ is empty, $C_1 = C_2 = \varnothing$ fulfill the conditions of this lemma.
    
    If on the other hand, $G$ consists of exactly two components, we notice that $G \in \free_\prec(P_3)$, since the endpoints of one $P_3$ in one component combined with any vertex of the other component would induce $\overline{K_3}$, a contradiction.
    Thus, $G$ is a cluster graph and setting $C_1, C_2$ to the vertex set of one component each fulfills the conditions of this lemma.
    Thus, from now on, we will assume that $G$ is non-empty and consists of exactly one connected component.

    We proceed with deducing the precise structure of $G$ from our premises.
    Notice that $G[V \setminus M]$ is $P_3$-free, since $P_3 \prec G[V \setminus M]$ implies that $M \cap (V \setminus M) \neq \varnothing$, a contradiction.

    Since $\overline{K_3} \prec G[V \setminus M]$ would imply $\overline{K_3} \prec G$, we obtain that $G[V \setminus M] \in \free_\prec(\set{P_3, \overline{K_3}})$, i.e., it is a cluster graph of at most two clusters.
    We henceforth use $\mathcal{C}$ to denote the vertex sets of all connected components of $G[V \setminus M]$.

    Furthermore, we derive $M \neq V$, since if $M = V$, $G$ would be a non-empty clique because $G[M]$ is, yet, since $M$ would be empty, $G$ would also be an empty graph, a contradiction.
    
    We can use $M \neq V$ to show $\set{v_1v_2 \mid v_1 \in M, v_2 \in C} \subseteq E$ for all $C \in \mathcal{C}$ such that $G[M \cup C]$ is connected:
    Let $C \in \mathcal{C}$ such that $G[M \cup C]$ is connected.
    Since $C \neq M$ (because $V \neq M$) and $M \neq \varnothing$, we can select $u \in C$, such that $uv \in E$, where $v \in M$. 
    Suppose there is $w \in C$, such that $vw \not\in E$.
    Since $G[C]$ is a clique, we know that $uw \in E$.
    Thus, $\set{u, v, w}$ induce $P_3$ in $G$ and its middle point $u$ is an element of $M$, a contradiction to $C \subseteq V \setminus M$;
    see \cref{figure:lemma_p3_midpoints_induce_scc} for an illustration.
    Thus, it is indeed the case that $\set{v_1v_2 \mid v_1 \in M, v_2 \in C} \subseteq E$ for all $C \in \mathcal{C}$ such that $G[M \cup C]$ is connected.

    Since $G[V \setminus M]$ is non-empty, we know that $|\mathcal{C}| \ge 1$. Also, $|\mathcal{C}| \leq 2$, since if $|\mathcal{C}| \ge 3$, we would obtain $G \prec \overline{K_3}$. We will now show that $|\mathcal{C}| = 2$ by deriving an absurdity from the other remaining possibility:

    Towards a contradiction, suppose $G[V \setminus M]$ consists of exactly one connected component, that is $\mathcal{C} = \set{V \setminus M}$.
    Since $G$ is connected, $V \neq M$, and $M \neq \varnothing$, we can select $u \in G[V \setminus M]$ such that $uv \in E$ where $v \in M$.
    Because $v$ is the middle-point of a $P_3$ in $G$, there is $w \in V \setminus \set{u, v}$ such that $vw \in E$ and $uw \not\in E$.

    \looseness=-1
    Suppose that $w \in M$. Since $\set{v_1v_2 \mid v_1 \in M, v_2 \in V \setminus M} \subseteq E$, we then have $uw \in E$, contradicting $uw \not\in E$.
    Now, suppose the opposite, i.e.,\ $w \in V \setminus M$. Thus, $w$ and $u$ are both part of the same connected component $G[V \setminus M]$. But this component is a clique, hence $uw \in E$, contradicting $uw \not\in E$.
    Therefore, in total, we conclude that $|\mathcal{C}| = 2$; we denote its elements by $C_1$ and $C_2$ and claim that $\set{C_1 \cup M, C_2 \cup M}$ a sigma clique cover of the desired properties.

    To prove this claim, we check both conditions of \cref{definition:sccshortversion}.
    For the first condition, we need to establish that $C_1 \cup M$ and $C_2 \cup M$ both induce cliques in $G$: Without loss of generality, we only consider $C_1 \cup M$. Since $G[C_1 \cup M]$ is connected, we know that $\set{v_1v_2 \mid v_1 \in M, v_2 \in C_1} \subseteq E$.
    By precondition, we have that $G[M]$ is a clique.
    Also, $G[C_1]$ is a clique. Hence, we conclude that $G[C_1 \cup M]$ is a clique too.

    Now, we need to establish that all edges of $G$ are covered by our supposed sigma clique cover. Let $v_1v_2 \in E$.
    If $v_1, v_2 \in C_1 \cup M$ or $v_1, v_2 \in C_2 \cup M$, then $v_1v_2$ is covered because $G[C_1 \cup M]$ and $G[C_2 \cup M]$ are cliques.
    Since $(C_1 \cup M) \cup (C_1 \cup M) = V$, only the case (without loss of generality) $v_1 \in C_1$ and $v_2 \in C_2$ is left to consider.
    If $v_1 \in M$ (resp. $v_2 \in M$), then $v_1 \in C_2$ (resp. $v_2 \in C_1$) and both vertices are covered by the clique $G[C_2 \cup M]$ (resp. $G[C_1 \cup M]$).
    Otherwise, $v_1, v_2 \in G[V \setminus M]$. But then $G[C_1]$ and $G[C_2]$ are connected, thus $|\mathcal{C}| \neq 2$, a contradiction to $|\mathcal{C}| = 2$.

    Using these premises, we can also establish the required condition on the weight.
    Since $M \cap C_1 = \varnothing$, $M \cap C_2 = \varnothing$, $C_1 \cap C_2 = \varnothing$, and $\set{C_1, C_2}$ is a partition of $V \setminus M$, we obtain $\wgt(\set{C_1 \cup M, C_2 \cup M}) = |C_1| + |C_2| + 2|M| = |V \setminus M| + 2|M| = |V| + |M|$.

    Finally, we need to show that $(C_1 \cup M) \cap (C_2 \cup M) = M$, which can be done in a direct manner: $(C_1 \cup M) \cap (C_2 \cup M) = (C_1 \cap (C_2 \cup M)) \cup (M \cap (C_2 \cup M))$ = $\varnothing \cup M = M$.
\end{proof}
This now yields a polynomial-time algorithm for \pFVS{\{P_3, \overline{K_3}\}}: Check whether the input graph consists of two cliques that overlap in at most $k$ vertices.
This finishes the outline of our characterization and we obtain:
\begin{restatable}[\appsymb]{theorem}{theoremleqthreecomplexity}\label{theorem:leq3_complexity}
  Let $\mathcal{F}$ be a set of graphs containing graphs of at most three vertices each,
  then \textsc{$\free_\prec(\mathcal{F})$-Vertex Splitting} is \NP-complete if $\mathcal{F} = \set{P_3}$ or $\mathcal{F} = \set{K_3}$ and admits a polynomial-time algorithm otherwise.
\end{restatable}
\looseness=-1
The polynomial-time results for split- and threshold graphs use the observation that destroying some of their forbidden subgraphs by splitting, namely $P_4$, $C_4$, or $C_5$, necessarily creates another forbidden subgraph $\overline{C_4}$, reducing the problem to checking whether the input graph has the respective property.
This seems to be a general principle worthy of further exploration.

\toappendix{

  \section{Properties Characterized by Small Forbidden Induced Subgraphs: Full Proof}
  \label{section:leq3-full}

In this section, we provide a dichotomy regarding the classical complexity of \textsc{$\Pi$-Vertex Splitting} for properties characterized by sets of forbidden induced subgraphs, each containing no more than three vertices.
For each such property $\Pi$, we either demonstrate that \textsc{$\Pi$-Vertex Splitting} is in \P{}, or show that the problem is \NP-complete (\cref{theorem:leq3_complexity}).

The set of graphs of at most three vertices is given by $\set{K_0, K_1, K_2, \overline{K_2}, K_3, \overline{K_3}, P_3, \overline{P_3}}$.
Hence, we need to cover $2^8 - 2 = 254$ cases.
The following is an attempt to do so using a minimal number of lemmas, each dealing with a set of \textsc{$\Pi$-Vertex Splitting} problems of structurally similar $\Pi$.

The majority of cases will be covered by dealing with sets of forbidden induced subgraphs such that,
as soon as one forbidden subgraph is present, it is impossible to reach the desired graph class via vertex splitting (\cref{section:indestructible_fisgs}).
Afterward, three more involved cases remain.
Of these three, two concern polynomial-time solvable restrictions of the \NP-complete \textsc{Cluster Vertex Splitting} problem (\cref{section:poly_cvs}),
and finally, the last case is solved using a short excursion to Ramsey Theory (\cref{section:poly_cvs}).
We integrate all of these cases into a complete dichotomy in \cref{section:dichotomy}.

Finally, we show that \textsc{Split-} and \textsc{Threshold-Vertex Splitting} admit polynomial-time algorithms in \cref{section:split_thres}.

Throughout the section, we will commonly make use of the fact
that a graph property characterized by a finite set of forbidden induced subgraphs
can be recognized in polynomial time:
\begin{proposition}\label{lemma:k_is_zero_polynomial}
    Let $G$ be a graph and let $\mathcal{F}$ be a fixed, finite set of graphs.
    Then,
    the instance $(G, 0)$ of \textsc{$\free_\prec(\mathcal{F})$-Vertex Splitting} 
    can be decided in polynomial time.
\end{proposition}
\begin{proof}

We observe that $(G, 0)$ is a positive instance of our problem if and only if $G \in \free_\prec(\mathcal{F})$.
Since $\free_\prec(\mathcal{F}) = \bigcap_{H \in \mathcal{F}} \free_\prec(\set{H})$ and $\mathcal{F}$ is finite, the problem is reduced to checking each forbidden induced subgraph individually.
This can be accomplished with the following brute-force approach:
Let $H \in \mathcal{F}$.
Any $X \subseteq V(G)$ with $G[X] \simeq H$ is of size $|V(H)|$; thus, there are $\binom{|V(G)|}{|V(H)|} \in \mathcal{O}(|V(G)|^{|V(H)})$ candidates to consider for $X$.
A given candidate set $X$ can be checked as follows: For each of the possible $|V(H)|!$ permutations of $X$, build an incidence matrix of $G[X]$ with respect to the current ordering and compare it to a fixed incidence matrix of $H$. Crucially, $|V(H)|$ is constant.

Thus, in total, the running time of the complete procedure is bounded by a polynomial in $|V(G)|$. %
\end{proof}

Also, the class $\free_\prec(\set{P_3, K_3})$ will arise frequently.
Hence, it will be convenient to have a simple description of the graphs that constitute this class:
\begin{lemma}\label{lemma:k3_p3_characterization_1}
    Let $G$ be a graph. Then, $G \in \free_\prec(\set{P_3, K_3})$ if and only if
    each component of $G$ is composed of at most two vertices.
\end{lemma}
\begin{proof}
    We observe that $\free_\prec(\set{P_3, K_3}) = \free_\prec(\set{P_3}) \cap \free_\prec(\set{K_3})$.
    In other words, $G$ is a cluster graph that is also triangle-free.
    Since all cliques of size at least three contain triangles as induced subgraphs, all connected components of $G$ do not contain more than two vertices.
\end{proof}

\subsection{Indestructible Forbidden Induced Subgraphs}
\label{section:indestructible_fisgs}

In this section, we address sets of forbidden induced subgraphs $\mathcal{F}$ that are \emph{indestructible}.
We say a set of forbidden induced subgraphs $\mathcal{F}$ is indestructible if it satisfies the following condition: Given any splitting sequence $G_0, \dots, G_\ell$, if a graph $G_i$ does not belong to $\free_\prec(\mathcal{F})$, then none of the subsequent graphs $G_j$ with $j \geq i$ belong to $\free_\prec(\mathcal{F})$ either.

Thus, each instance $(G, k)$ of \textsc{$\free_\prec(\mathcal{F})$-Vertex Splitting}
is equivalent to the instance $(G, 0)$,
meaning it suffices to determine whether $G \in \free_\prec(\mathcal{F})$.

\subsection{Cluster Graphs With Clusters of at Most Two Vertices as Forbidden Induced Subgraphs}

If a graph $G$ contains a certain number of isolated vertices, call it $i$, and a matching of, say, $m$ edges,
any graph obtainable from $G$ via vertex splitting will also contain at least $i$ isolated vertices and a matching of size at least $m$.
To show this, we introduce two simple lemmas.
The first lemma states that non-edges are preserved when splitting a vertex:
\begin{lemma}\label{lemma:splits_maintain_independent_pairs}
Let $G$ be a non-empty graph and $G'$ be obtained from $G$ by splitting some vertex $w \in V(G)$ into $w_1, w_2 \in V(G')$; furthermore let $v_1,v_2 \in V(G)$.
If $v_1v_2 \not\in E(G)$, then for all descendants $v_1', v_2' \in V(G')$ of $v_1$ and $v_2$ respectively, it holds that $v_1'v_2' \not\in E(G')$.
\end{lemma}
\begin{proof}
    \emph{Case $w \not\in \set{v_1, v_2}$}:
    Neither $v_1$ nor $v_2$ are split in this case, so the descendants $v_1', v_2'$ of $v_1, v_2$ are uniquely determined: $v_1' = v_1$ and $v_2' = v_2$.
    The neighborhood of $v_1$ changes only insofar, as that $w$ is possibly exchanged for some subset of $\set{w_1, w_2}$, thus
    \begin{equation*}
        N_{G'}(v_1') \subseteq \left ( N_G(v_1) \setminus \set{w} \right ) \cup \set{w_1, w_2}.
    \end{equation*}
    Since $v_2' \not\in N_G(v_1)$ and $v_2' \not\in \set{w_1, w_2}$, we conclude that $v_2'$ is not
    a member of the superset on the right-hand side, therefore, it is also not included in $N_{G'}(v_1')$. Thus, we obtain $v_1'v_2' \not\in E(G')$.

    \smallskip
    \emph{Case $w \in \set{v_1, v_2}$}:
    Without loss of generality, $w = v_1$.
    Therefore, for any descendants $v_1', v_2'$ of $v_1, v_2$ respectively, it holds that $v_1' \neq v_1$ and $v_2' = v_2$.
    Vertices that are split maintain a subset of their original neighborhood; in our case this means
    \begin{equation*}
        N_{G'}(v_1') \subseteq N_G(v_1).
    \end{equation*}
    Since $v_2' \not\in N_G(v_1)$, we obtain $v_2' \not\in N_{G'}(v_1')$, thus $v_1'v_2' \not\in E(G')$.
\end{proof}
There is an analog for edges as well.
In this case, only at least one descendant edge is preserved, in contrast to the case of non-edges, where all descendant non-edges are preserved.
\begin{lemma}\label{lemma:splits_maintain_edges}
    Let $G$ be a non-empty graph and $G'$ be obtained from $G$ by splitting some vertex $w \in V(G)$ into $w_1, w_2 \in V(G')$; furthermore let $v_1,v_2 \in V(G)$.
    If $v_1v_2 \in E(G)$, then there are two descendants $v_1', v_2' \in V(G')$ of $v_1$ and $v_2$ respectively, such that $v_1'v_2' \in E(G')$.
\end{lemma}
\begin{proof}
    \emph{Case $w \not\in \set{v_1, v_2}$}:
    In this case, neither $v_1$ nor $v_2$ is split. By definition of vertex splitting, it holds that
    \begin{equation*}
        N_{G'}(v_1) = \left( N_G(v_1) \setminus \set{w} \right) \cup W,
    \end{equation*}
    where $W \subseteq \set{w_1, w_2}$.
    Since $v_2 \in N_G(v_1)$ but $v_2 \not\in W \cup \set{w}$, it follows that
    $v_2 \in N_{G'}(v_1)$.
    Thus, we obtain $v_1v_2 \in E(G')$ and observe that $v_1$ originates from itself, as does $v_2$.

    \smallskip
    \emph{Case $w \in \set{v_1, v_2}$}:
    Without loss of generality, we may assume that $w = v_1$, that is, $v_1$ is split.
    By definition of vertex splitting, we have
    \begin{equation*}
        N_{G'}(w_1) \cup N_{G'}(w_2) = N_G(v_1).
    \end{equation*}
    Thus, since $v_2 \in N_G(v_1)$, either $v_2 \in N_{G'}(w_1)$, or $v_2 \in N_{G'}(w_2)$, or both.
    Without loss of generality, we assume the former case.
    We observe that $w_1v_2 \in E(G')$ and that $v_2$ originates from itself, whereas $w_1$ originates from $v_1$.
\end{proof}

With these two lemmas, we prove that sets of graphs of the class in question are indestructible:
\begin{lemma}\label{lemma:splits_maintain_cluster_graphs_with_tiny_clusters}
    Let $G$ be a graph, $H \in \free_\prec(\set{K_3, P_3})$, and $H \prec G$.
    Then, $H \prec G'$ for any $G'$ obtainable from $G$ via a vertex split. 
\end{lemma}
\begin{proof}
    Let $X \subseteq V(G)$ such that $G[X] \simeq H$.
    We construct a new set $X' \subseteq V(G')$ for which $G'[X'] \simeq H$ will hold as follows:
    Map all vertices of $X$ that are isolated in $G[X]$ to any of their descendant vertices in $G'$,
    and map all vertex pairs of $X$ that are adjacent in $G[X]$ to some adjacent pair of their descendant vertices, as is possible by \cref{lemma:splits_maintain_edges}.
    Using the characterization given in \cref{lemma:k3_p3_characterization_1}, we see that this suffices to map each connected component of $G[X]$, that is, either an isolated vertex or an isolated edge,
    to either a distinct vertex or a distinct edge in $G'[X']$, respectively.

    Using \cref{lemma:splits_maintain_independent_pairs}, we observe that each edge not present in $G[X]$ forces that the corresponding edge (composed of the two corresponding descendant vertices) is also non-existent in $G'[X']$.
    Ergo, the number of connected components of $G[X]$ equals the number of connected components of $G'[X']$, implying $G'[X'] \simeq H$.
\end{proof}

Now we immediately obtain a polynomial-time algorithm:

\begin{proposition}\label{lemma:splitting_polynomial_for_cluster_graphs_at_most_two}
    Let $\mathcal{F} \subseteq \free_\prec(\set{K_3, P_3})$ where $\mathcal{F}$ is finite.
    Then, \textsc{$\free_\prec(\mathcal{F})$-Vertex Splitting} admits a polynomial-time algorithm. 
\end{proposition}
\begin{proof}
  Let $(G, k)$ be an instance of \textsc{$\free_\prec(\mathcal{F})$-Vertex Splitting}.
  If $\F = \emptyset$ then the algorithm may always return yes.
  Otherwise, there is at least one graph in \F.
  By applying \cref{lemma:splits_maintain_cluster_graphs_with_tiny_clusters},
  we know that none of the forbidden induced subgraphs of $\mathcal{F}$ can be removed via splitting in any splitting sequence.
  Therefore, $(G, k)$ is a positive instance if and only if $(G, 0)$ is.
  This we can check in polynomial time by \cref{lemma:k_is_zero_polynomial}.
\end{proof}

\subsection{Forbidden Induced Subgraphs That When Destroyed Introduce a Forbidden Induced \texorpdfstring{$\overline{P_3}$}{P3-complement} }

By the last subsection, we know that $\set{\overline{P_3}}$ is indestructible.
In general, it is not necessary that a superset of an indestructible set is indestructible too.
But note that some graphs, for example $P_3$ and $K_3$, necessarily introduce an induced $\overline{P_3}$ when they get destroyed in a splitting sequence.
Thus, a set of indestructible graphs $\mathcal{F}$ containing $\overline{P}$, augmented with either $P_3, K_3$, or both, forms an indestructible set.
Hence, we obtain the following lemma:

\begin{proposition}\label{lemma:p3_complement_introduction}
    Let $\mathcal{F} \subseteq \free_\prec(\set{P_3, K_3})$ with $\overline{P_3} \in \mathcal{F}$ as well as $\mathcal{F}$ finite,
    and let $\mathcal{G} \subseteq \set{P_3, K_3}$ with $\mathcal{G} \neq \varnothing$.
    Then, \textsc{$\free_\prec(\mathcal{F} \cup \mathcal{G})$-Vertex Splitting} admits a polynomial-time algorithm. 
\end{proposition}
\begin{figure}
    \begin{center}
        \includegraphics{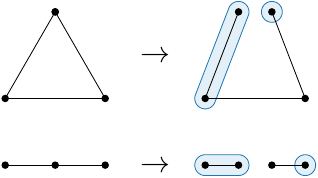}
    \end{center}
    \caption[The graphs $K_3$ and $P_3$ are split to break isomorphism.]{The graphs $K_3$ and $P_3$ are split to break isomorphism. An induced embedding of $\overline{P_3}$ is highlighted in each graph after the split.}
    \label{figure:K_3_compl_introduced}
\end{figure}
\begin{proof}
Let $(G, k)$ be an instance of \textsc{$\free_\prec(\mathcal{F})$-Vertex Splitting}.
If there is $H \in \mathcal{F}$ with $H \prec G$, then $(G, k)$ is a negative instance of 
\textsc{$\free_\prec(\mathcal{F} \cup \mathcal{G})$-Vertex Splitting}, since by \cref{lemma:splits_maintain_cluster_graphs_with_tiny_clusters},
$H$ cannot be destroyed by splitting vertices.
Since $\mathcal{F}$ is finite, this check can be performed in polynomial-time using \cref{lemma:k_is_zero_polynomial}.
Otherwise, again with \cref{lemma:k_is_zero_polynomial},
we can decide the instance $(G, 0)$ of \textsc{$\free_\prec(\mathcal{G})$-Vertex Splitting}.
If the result is positive, then so is the instance $(G, k)$ of \textsc{$\free_\prec(\mathcal{F} \cup \mathcal{G})$-Vertex Splitting}.
Otherwise, there is $H \in \mathcal{G}$ with $H \prec G$.

Suppose there is a splitting sequence $G_0, \dots, G_\ell$ with $G_0 = G$ such that $H \not\prec G_\ell$.
Then, there is $i \in \set{0, \dots, \ell-1}$ such that $H \prec G_i$ and $H \not\prec G_{i+1}$.
Irrespective of whether $H = P_3$ or $H = K_3$, to destroy the copy of $H$ in $G_i$, a vertex that has degree two in the copy must have been split in $G_i$ to produce $G_{i+1}$.
Also, the two edges incident to the split vertex in the copy must have been assigned to different descendants, for otherwise, the copy would persist.
But then, the split introduces a new induced $\overline{P_3}$ in $G_{i+1}$, which cannot be removed via vertex splitting by \cref{lemma:splits_maintain_cluster_graphs_with_tiny_clusters};
see \cref{figure:K_3_compl_introduced} for an illustration.
Hence, we have derived a contradiction and can conclude that the original instance is negative.
\end{proof}

\subsection{Properties With a Forbidden Induced Subgraph of Size at Most Two}
The last observation that we make in this section is that when a finite set of forbidden induced subgraphs contains a graph of at most two vertices, then
the associated vertex splitting problem becomes trivial:

\begin{proposition}\label{lemma:splitting_polynomial_for_leq2}
    Let $\mathcal{F}$ be a finite set of graphs with $\set{K_0, K_1, K_2, \overline{K_2}} \cap \mathcal{F} \neq \varnothing$.
    Then, \textsc{$\free_\prec(\mathcal{F})$-Vertex Splitting} admits a polynomial-time algorithm. 
\end{proposition}
\begin{proof}
We perform a case analysis.

\smallskip
\emph{Case $K_0 \in \mathcal{F}$}:
The empty graph is an induced subgraph of all graphs; hence $(G, k)$ is a negative instance.

\smallskip
\emph{Case $K_1 \in \mathcal{F}$}:
If $K_1 \prec G$, then $(G, k)$ is a negative instance since vertex splitting never reduced the number of vertices.
If otherwise, $G$ does not contain vertices, it cannot be split.
Hence $(G, k)$ is a positive instance if and only if $(G, 0)$ is.
This can be decided in polynomial time (\cref{lemma:k_is_zero_polynomial}).

\smallskip
\emph{Case $K_2 \in \mathcal{F}$}:
If $K_2 \prec G$, then $(G, k)$ is a negative instance, for $K_2$ cannot be removed via vertex splitting by \cref{lemma:splits_maintain_edges}.
Otherwise, $G$ is edge-less and all induced subgraphs of $G$ are independent sets.
Through splitting, only more independent sets of higher cardinality can be introduced, but none removed (\cref{lemma:splits_maintain_cluster_graphs_with_tiny_clusters}).
Hence, $(G, k)$ is a positive instance if and only if $(G, 0)$ is. This can be decided in polynomial time (\cref{lemma:k_is_zero_polynomial}).

\smallskip
\emph{Case $\overline{K_2} \in \mathcal{F}$}:
Similarly, if $\overline{K_2} \prec G$, then $(G, k)$ is a negative instance, for $\overline{K_2}$ cannot be removed via vertex splitting by \cref{lemma:splits_maintain_independent_pairs}.
Otherwise, $G$ contains at least one edge. Since any split introduces a new $\overline{K_2}$ as the two descendants of the vertex that is split are always independent, $(G, k)$ is a positive instance if and only if $(G, 0)$ is.
This can be decided in polynomial time (\cref{lemma:k_is_zero_polynomial}).
\end{proof}

\subsection[Two tractable Restrictions of \textsc{Cluster Vertex Splitting}]{Two Polynomial-Time Solvable Restrictions of \textsc{Cluster Vertex Splitting}}
\label{section:poly_cvs}

Firbas et al.~\cite{firbas_cluster_2023} showed that \textsc{Cluster Vertex Splitting} (\textsc{CVS}), that is, \textsc{$\free_\prec(P_3)$-Vertex Splitting}, is \NP-complete.
Here, we study two restrictions of said problem that render it polynomial-time solvable.

\subsection{Cluster graphs with clusters of at most two vertices}

\begin{table}
    \caption[A graph is transformed to a member of $\free_\prec(\set{P_3, K_3})$.]{A minimal-length splitting sequence transforming $K_3 \cupdot K_1$ into a graph belonging to $\free_\prec(\set{P_3, K_3})$.
    For each graph, the number of non-isolated vertices is counted, as well as the number of edges multiplied by two.
    These two numbers are equal in the final graph, and thus, according to \cref{lemma:k3_p3_characterization_2}, it is a member of the desired class.
    }
    \label{table:splitting_seq_table}
    \begin{center}
        \begin{tabular}{@{}c Sc Sc Sc Sc@{}} \toprule
            & \cincludegraphics{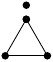} & \cincludegraphics{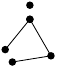} & \cincludegraphics{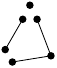} & \cincludegraphics{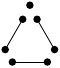} \\ \midrule
           $|V \setminus I|$ & 3 & 4 & 5 & 6 \\
           $2|E|$ & 6 & 6 & 6 & 6 \\
           \bottomrule
       \end{tabular}
    \end{center}
\end{table}

The first restriction of \textsc{Cluster Vertex Splitting} deals with cluster graphs consisting only of isolated vertices and isolated edges, that is, the class $\free_\prec(\set{K_3, P_3})$.
Intuitively, the optimal strategy to turn a graph into this shape is to 
``split away'' edge after edge, until all edges are isolated.
Consult \cref{table:splitting_seq_table} for an example.

We characterize the class in question in terms of an equation relating the number of edges, isolated vertices, and non-isolated vertices.
Then, using this characterization, we will be able to deduce the optimal strategy to solve the associated Vertex Splitting problem.

\begin{lemma}\label{lemma:k3_p3_characterization_2}
    Let $G = (V, E)$ be a graph and let $I \coloneqq \set{v \in V \mid d_G(v) = 0}$, that is, all isolated vertices of $G$.
    Then, $G \in \free_\prec(\set{K_3, P_3})$ if and only if $2|E| = |V \setminus I|$.
\end{lemma}
\begin{proof}
$(\Rightarrow)\colon$
Let $G = (V, E) \in \free_\prec(\set{K_3, P_3})$. By \cref{lemma:k3_p3_characterization_1}
we get that $G[V \setminus I]$ is a disjoint union of $K_2$'s.
Furthermore, $E = E(G[V \setminus I])$. Hence, $2|E| = |V \setminus I|$.

$(\Leftarrow)\colon$
Let $G = (V, E)$ be a graph with $2|E| = |V \setminus I|$.
By the handshaking lemma, we get
\begin{DispWithArrows*}[fleqn,mathindent=25pt,displaystyle,wrap-lines]
    |V \setminus I| &= \sum_{v \in V} d_G(v)  \Arrow{$d_G(v) = 0$ for all $v \in I$}\\
                    &= \sum_{v \in V \setminus I} d_G(v) \Arrow{$d_G(v) \ge 1$ for all $v \in V \setminus I$}\\
                    d_G(v) = 1 &\text{ for all } v \in V \setminus I.
\end{DispWithArrows*}

We conclude that all connected components that are not isolated are isomorphic to $K_2$.
Thus, by \cref{lemma:k3_p3_characterization_1}, we derive $G \in \free_\prec(\set{K_3, P_3})$.
\end{proof}

Consider a graph that does not fulfill the characterization given in the last lemma.
Then, we have $|V \setminus I| < 2|E|$. To transform this inequality to an equality using vertex splitting,
we need to increase the left-hand side of the inequality, that is,  $|V \setminus I|$, while not increasing the right-hand side, that is, $2|E|$.
The next lemma states that this is always possible.

\begin{lemma}\label{lemma:split_preserving_edge_cardinality}
    Let $G$ be a graph and $I \subseteq V(G)$ the set of isolated vertices of $G$.
    If $|V(G) \setminus I| < 2|E(G)|$, then there is a vertex split producing $G'$ such that
    $G'$ has the same number of edges as $G$, and the same set of isolated vertices as~$G$.
\end{lemma}
\begin{proof}
    Using an argument like in the proof of \cref{lemma:k3_p3_characterization_2}, we deduce that
    \begin{equation*}
        |V(G) \setminus I| < \sum_{v \in V \setminus I} d_G(v).
    \end{equation*}
    Since we are summing over non-isolated vertices, we observe that there is $v \in V \setminus I$ with $d_G(v) \ge 2$.
    Now, each non-trivial disjoint split of $v$ yields a graph satisfying the conditions of this lemma.
\end{proof}

\begin{algorithm}[t]
    \SetKwInput{KwInput}{Input}
    \SetKwInput{KwOutput}{Output}
    \DontPrintSemicolon 
    \SetAlgoLined
    \KwInput{$(G, k)$, instance of \textsc{$\free_\prec(\set{P_3, K_3})$-Vertex Splitting}}
    \KwOutput{\texttt{true} if $G$ can be made a member of $\free_\prec(\set{P_3, K_3})$ using at most $k$ vertex splits, \texttt{false} otherwise}
    \BlankLine
    $I \gets \set{v \in V(G) \mid d_G(v) = 0}$\;
    \Return{$k \geq 2\M{G} + |I| - \N{G}$}
    \caption{\textsc{$\free_\prec(\set{P_3, K_3})$-Vertex Splitting}}
    \label{algorithm:p3k3_vertex_splitting}
\end{algorithm}

Finally, we make our reasoning rigorous and derive an appropriate algorithm:
\begin{proposition}\label{lemma:p3_k3_splitting_polynomial}
    There is a polynomial-time algorithm for \textsc{$\free_\prec(\set{P_3, K_3})$-Vertex Splitting}.
\end{proposition}
\begin{proof}
    Let $(G, k)$ with $G = (V, E)$ be an instance of \textsc{$\free_\prec(\set{P_3, K_3})$-Vertex Splitting}.
    Furthermore, let $I \subseteq V$ be the set of isolated vertices of $V$.
    We claim that \cref{algorithm:p3k3_vertex_splitting} decides the problem correctly,
    that is,
    $(G, k)$ is a positive instance of \textsc{$\free_\prec(\set{P_3, K_3})$-Vertex Splitting} if and only if $k \geq 2|E| + |I| - |V|$.
    To show the correctness of \cref{algorithm:p3k3_vertex_splitting}, 
    we prove that there is a splitting sequence of length $2|E| + |I| - |V|$ such that the last graph is in $\free_\prec(\set{P_3, K_3})$, and that there is no shorter such sequence.
    
    \smallskip\emph{Existence:}
    By applying \cref{lemma:split_preserving_edge_cardinality} repeatedly for as long as it is applicable,
    we obtain a splitting sequence where the number of edges, as well as the number of isolated vertices, remains invariant throughout the sequence, while the number of vertices increases by one for each subsequent graph.
    This process cannot continue indefinitely and for the last graph of the sequence, call it $G_\ell$, we find that $|V(G_\ell) \setminus I| = 2|E(G_\ell)|$.
    Using \cref{lemma:k3_p3_characterization_2}, we conclude that $G_\ell$ is a member of $\free_\prec(\set{P_3, K_3})$.
    Since the left-hand side of the inequality in \cref{lemma:split_preserving_edge_cardinality} increases by one for each graph in the sequence,
    while the right-hand side remains constant, we can deduce via subtraction that $\ell = 2|E| + |I| - |V|$.

    \smallskip\emph{Minimality:}
    Let $G_0, \dots, G_\ell$ be a minimum-length sequence of graphs generated by successive vertex splits with $G_0 = G$
    such that $G_\ell \in \free_\prec(\set{K_3, P_3})$.
    The sequence does not introduce any new isolated vertices, for if it did, a shorter sequence satisfying our conditions could be obtained by removing such operations.

    We note that $|E(G_\ell)| \ge |E|$ since the number of edges in a graph cannot decrease by vertex splitting. 
    Moreover, $|V(G_\ell)| = |V| + \ell$, since each of the $\ell$ splits introduces exactly one new vertex.
    By \cref{lemma:k3_p3_characterization_2}, it holds that $2|E(G_\ell)| = |V(G_\ell)| - |I|$.
    Using these premises, we obtain
    \begin{equation*}
        2|E| \leq 2|E(G_\ell)| = |V| + \ell - |I|. 
    \end{equation*}
    Finally, by subtracting $|V| - |I|$ from both sides, we conclude that $\ell \ge 2|E| + |I| - |V|$.
\end{proof}

\subsection{Cluster Graphs of at Most Two Clusters}

Now, instead of cluster graphs with clusters of size at most two, we consider cluster graphs of arbitrarily sized clusters, but the restriction that there shall be at most two of them.
This class is described by $\free_\prec(\set{P_3, \overline{K_3}})$.
Intuitively, to solve \textsc{$\free_\prec(\set{P_3, \overline{K_3}})$-Vertex Splitting}, we need to recognize graphs that consist of at most two possibly overlapping clusters, such that their overlap spans at most $k$ vertices.
See \cref{figure:duo_cluster} for an example.

To capture the locations of the clusters (cliques) of the solution cluster graph, 
we use the notion of sigma clique covers:
\begin{definition}\label{definition:scc}
  Let $G$ be a graph. Then, $\mathcal{C} \subseteq \mathcal{P}(V)$ is called a \emph{sigma clique cover} of 
  $G$ if
  \iflong    
    \begin{enumerate}
    \item $G[C]$ is a clique for all $C \in \mathcal{C}$ and
    \item for each $e \in E(G)$, there is $C \in \mathcal{C}$ such that $e \in E(G[C])$, that is, all edges of $G$ are ``covered'' by some clique of $\mathcal{C}$.
    \end{enumerate}
    The \emph{weight} of a sigma clique cover $\mathcal{C}$ is denoted by $\wgt(\mathcal{C})$, where
    \[
      \wgt(\mathcal{C}) \coloneqq \sum_{C \in \mathcal{C}} |C|.    
    \]  

  \else
    \begin{inparaenum}
    \item $G[C]$ is a clique for all $C \in \mathcal{C}$ and
    \item for each $e \in E(G)$, there is $C \in \mathcal{C}$ such that $e \in E(G[C])$, that is, all edges of $G$ are ``covered'' by some clique of $\mathcal{C}$.
    \end{inparaenum}
    The \emph{weight} of a sigma clique cover $\mathcal{C}$ is $\wgt(\mathcal{C}) \coloneqq \sum_{C \in \mathcal{C}} |C|$.
  \fi      
\end{definition}
Intuitively, the sets in the sigma clique cover correspond to the clusters of the solution cluster graph.
The following lemma captures this relation: 
\begin{lemma}[Firbas et al.~\cite{firbas_cluster_2023}, Lemma 4.3]\label{lemma:cvs_scc_reduction}
  Let $G = (V,E)$ be a graph, and let $I \coloneqq \{v \in V \mid d_G(v) = 0\}$.
  Then, there are at most $k$ vertex splits that turn $G$ into a cluster graph if and only if $G$ admits a sigma clique cover with weight at most $|V| - |I| + k$.
\end{lemma}
To obtain a polynomial-time algorithm, we develop necessary and sufficient conditions for a graph 
to admit a sigma clique cover of size at most two and of certain weight.

In the next lemma, we essentially prove the following:
Consider the set of induced $P_3$ in a graph of independence number at most two.
If the midpoints of all induced $P_3$'s form a clique, then 
we are able to extract a sigma clique cover of size two from the graph.

\begin{figure}
    \begin{center}
        \includegraphics{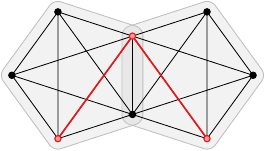}
    \end{center}
    \caption[Example for a positive instance \textsc{$\free_\prec(\set{P_3, \overline{K_3}})$-Vertex Splitting}.]{A graph $G$ for which $(G, 2)$ is a positive instance of \textsc{$\free_\prec(\set{P_3, \overline{K_3}})$-Vertex Splitting}.
    A sigma clique cover of size two and weight ten is highlighted in gray. One of the $3 \cdot 2 \cdot 3 = 18$ induced $P_3$'s is marked in red.
    The set of midpoints for all induced $P_3$'s is given by the intersection of the two cliques. Note that this set induces a clique of size two.
    }
    \label{figure:duo_cluster}
\end{figure}

\begin{lemma}\label{lemma:p3_midpoints_induce_scc}
    Let $G = (V, E) \in \free_\prec(\set{\overline{K_3}})$ without isolated vertices
    and let $M \coloneqq \set{ v \in V \mid \exists f \in \emb_\prec(P_3, G) \colon v \in f(V(P_3)) \land d_{P_3}(f^{-1}(v)) = 2}$, that is, the set of all vertices in $G$ that are a midpoint of some induced $P_3$ in $G$. If $G[M]$ is a non-empty complete graph, then there are $C_1, C_2 \subseteq V$ such that $\set{C_1, C_2}$ is a sigma clique cover of $G$ with $C_1 \cap C_2 = M$ and $\wgt(\set{C_1, C_2}) = |V| + |M|$.
\end{lemma}
\begin{proof}
    First, we get some trivial cases out of the way.
    The graph $G$ cannot have more than two connected components, for then we would have $\overline{K_3} \prec G$. If $G$ is empty, $C_1 = C_2 = \varnothing$ fulfill the conditions of this lemma.
    
    If on the other hand, $G$ consists of exactly two components, we notice that $G \in \free_\prec(P_3)$, since the endpoints of one $P_3$ in one component combined with any vertex of the other component would induce $\overline{K_3}$, a contradiction. Thus, $G$ is a cluster graph and setting $C_1, C_2$ to the vertex set of one component each fulfills the conditions of this lemma.
    Thus, from now on, we will assume that $G$ is non-empty and consists of exactly one connected component, i.e., it is connected.

    We proceed with deducing the precise structure of $G$ from our premises.
    Notice that $G[V \setminus M]$ is $P_3$-free, since $P_3 \prec G[V \setminus M]$ implies that $M \cap (V \setminus M) \neq \varnothing$, a contradiction.

    Since $\overline{K_3} \prec G[V \setminus M]$ would imply $\overline{K_3} \prec G$, we obtain that $G[V \setminus M] \in \free_\prec(\set{P_3, \overline{K_3}})$, i.e., it is a cluster graph of at most two clusters.
    We henceforth use $\mathcal{C}$ to denote the vertex sets of all connected components of $G[V \setminus M]$.

    Furthermore, we derive $M \neq V$, since if $M = V$, $G$ would be a non-empty clique because $G[M]$ is, yet, since $M$ would be empty, $G$ would also be an empty graph, a contradiction.
    
    We can use $M \neq V$ to show $\set{v_1v_2 \mid v_1 \in M, v_2 \in C} \subseteq E$ for all $C \in \mathcal{C}$ such that $G[M \cup C]$ is connected:
    Let $C \in \mathcal{C}$ such that $G[M \cup C]$ is connected.
    Since $C \neq M$ (because $V \neq M$) and $M \neq \varnothing$, we can select $u \in C$, such that $uv \in E$, where $v \in M$. 
    Suppose there is $w \in C$, such that $vw \not\in E$.
    Since $G[C]$ is a clique, we know that $uw \in E$.
    Thus, $\set{u, v, w}$ induce $P_3$ in $G$ and its middle point $u$ is an element of $M$, a contradiction to $C \subseteq V \setminus M$;
    reference \cref{figure:lemma_p3_midpoints_induce_scc} for an illustration.
    Thus, it is indeed the case that $\set{v_1v_2 \mid v_1 \in M, v_2 \in C} \subseteq E$ for all $C \in \mathcal{C}$ such that $G[M \cup C]$ is connected.
    \begin{figure}
        \begin{center}
            \includegraphics{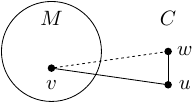}
        \end{center}
        \caption[Illustration used in the proof of \cref{lemma:p3_midpoints_induce_scc}.]{Illustration used in the proof of \cref{lemma:p3_midpoints_induce_scc}. The dashed line indicates a non-edge. }
        \label{figure:lemma_p3_midpoints_induce_scc}
    \end{figure}

    Since $G[V \setminus M]$ is non-empty, we know that $|\mathcal{C}| \ge 1$. Also, $|\mathcal{C}| \leq 2$, since if $|\mathcal{C}| \ge 3$, we would obtain $G \prec \overline{K_3}$. We will now show that $|\mathcal{C}| = 2$ by deriving an absurdity from the other remaining possibility:

    Towards a contradiction, suppose $G[V \setminus M]$ consists of exactly one connected component, that is $\mathcal{C} = \set{V \setminus M}$.
    Since $G$ is connected, $V \neq M$, and $M \neq \varnothing$, we can select $u \in G[V \setminus M]$ such that $uv \in E$ where $v \in M$.
    Because $v$ is the middle-point of a $P_3$ in $G$, there is $w \in V \setminus \set{u, v}$ such that $vw \in E$ and $uw \not\in E$.

    Suppose that $w \in M$. Since $\set{v_1v_2 \mid v_1 \in M, v_2 \in V \setminus M} \subseteq E$, we then have $uw \in E$, contradicting $uw \not\in E$.
    Now, suppose the opposite, i.e.,\ $w \in V \setminus M$. Thus, $w$ and $u$ are both part of the same connected component $G[V \setminus M]$. But this component is a clique, hence $uw \in E$, contradicting $uw \not\in E$.
    Therefore, in total, we conclude that $|\mathcal{C}| = 2$; we denote its elements by $C_1$ and $C_2$ and claim that $\set{C_1 \cup M, C_2 \cup M}$ a sigma clique cover of the desired properties.

    To prove this claim, we check both conditions of \cref{definition:scc}.
    For the first condition, we need to establish that $C_1 \cup M$ and $C_2 \cup M$ both induce cliques in $G$: Without loss of generality, we only consider $C_1 \cup M$. Since $G[C_1 \cup M]$ is connected, we know that $\set{v_1v_2 \mid v_1 \in M, v_2 \in C_1} \subseteq E$.
    By precondition, we have that $G[M]$ is a clique.
    Also, $G[C_1]$ is a clique. Hence, we conclude that $G[C_1 \cup M]$ is a clique too.

    Now, we need to establish that all edges of $G$ are covered by our supposed sigma clique cover. Let $v_1v_2 \in E$.
    If $v_1, v_2 \in C_1 \cup M$ or $v_1, v_2 \in C_2 \cup M$, then $v_1v_2$ is covered because $G[C_1 \cup M]$ and $G[C_2 \cup M]$ are cliques.
    Since $(C_1 \cup M) \cup (C_1 \cup M) = V$, only the case (without loss of generality) $v_1 \in C_1$ and $v_2 \in C_2$ is left to consider.
    If $v_1 \in M$ (resp. $v_2 \in M$), then $v_1 \in C_2$ (resp. $v_2 \in C_1$) and both vertices are covered by the clique $G[C_2 \cup M]$ (resp. $G[C_1 \cup M]$).
    Otherwise, $v_1, v_2 \in G[V \setminus M]$. But then $G[C_1]$ and $G[C_2]$ are connected, thus $|\mathcal{C}| \neq 2$, a contradiction to $|\mathcal{C}| = 2$.

    Using these premises, we can also establish the required condition on the weight.
    Since $M \cap C_1 = \varnothing$, $M \cap C_2 = \varnothing$, $C_1 \cap C_2 = \varnothing$, and $\set{C_1, C_2}$ is a partition of $V \setminus M$, we obtain $\wgt(\set{C_1 \cup M, C_2 \cup M}) = |C_1| + |C_2| + 2|M| = |V \setminus M| + 2|M| = |V| + |M|$.

    Finally, we need to show that $(C_1 \cup M) \cap (C_2 \cup M) = M$, which can be done in a direct manner: $(C_1 \cup M) \cap (C_2 \cup M) = (C_1 \cap (C_2 \cup M)) \cup (M \cap (C_2 \cup M))$ = $\varnothing \cup M = M$.
\end{proof}

Next, we prove that if a graph admits a sigma clique cover of two cliques that are incomparable with respect to the subset relation,
then the set of midpoints of all induced $P_3$ is given by the intersection of the two cliques.

\begin{lemma}\label{lemma:scc_induces_p3_midpoints}
    Let $G = (V, E)$ be a graph without isolated vertices and let
    $M \coloneqq \set{ v \in V \mid \exists f \in \emb_\prec(P_3, G) \colon v \in f(V(P_3)) \land d_{P_3}(f^{-1}(v)) = 2}$, that is, the set of all vertices in $G$ that are a midpoint of some $P_3 \prec G$.
    If $\set{C_1, C_2}$ 
    is a sigma clique cover of $G$
    with $C_1 \not\subseteq C_2$ and $C_2 \not\subseteq C_1$, then $M = C_1 \cap C_2$.
\end{lemma}
\begin{proof}
    We begin with deriving $M \subseteq C_1 \cap C_2$:
    Let $v \in M$. By $M$'s defining property, there are $u, w \in V$
    such that $uv, vw \in E$ but $uw \not\in E$.
    Towards a contradiction, assume that $v \not\in C_1 \cap C_2$.
    Since $G$ contains no isolated vertices, each vertex is covered by either $C_1$ or $C_2$.
    Without loss of generality, $v \in C_1$, but $v \not\in C_2$.
    Because $uw \not\in E$, $u$ and $w$ cannot be covered by the same clique of $\set{C_1, C_2}$.
    Thus, without loss of generality, $u \in C_1 \setminus C_2$ and $w \in C_2 \setminus C_1$.
    Consider the edge $vw \in E$: $vw$ is not covered by $C_1$ since $w \not\in C_1$, but $vw$ can neither be covered by $C_2$, since $v \not\in C_2$.
    This contradicts the fact that $\set{C_1, C_2}$ covers all edges of $G$.
    Thus, we have $v \in C_1 \cap C_2$ and $M \subseteq C_1 \cap C_2$.

    We will now show the converse, that is, $C_1 \cap C_2 \subseteq M$:
    Let $v \in C_1 \cap C_2$. Towards a contradiction, suppose $N_G(v) \subseteq C_1$. 
    Because $v$ is a member of the clique $C_2$, we obtain $C_2 \setminus \set{v} \subseteq N_G(v)$,
    and by our assumption, $C_2 \setminus \set{v} \subseteq C_1$ follows.
    By building the union with $\set{v}$ on both sides,
    we derive $C_2 \subseteq C_1$, contradicting our choice of $C_1$ and $C_2$.
    Thus, we have $N_G(v) \not\subseteq C_1$, that is, there is $w \in N_G(v) \setminus C_1$. Also, since $w$ must be covered by some clique (as $G$ is free of isolated vertices), $w \in C_2$ holds.
    By a completely symmetric argument, we additionally derive that there is $u \in N_G(v) \setminus C_2$ with $u \in C_1$.
    Thus, we conclude that $\set{u, w} \not\subseteq C_1$ and $\set{u, w} \not\subseteq C_2$, which
    forces $uw \not\in E$ since $\set{C_1, C_2}$ covers all edges of $G$.
    Remembering that $uv, vw \in E$, we conclude that $G[\set{u, v, w}] \simeq P_3$.
    As $v$ acts as the midpoint in this induced $P_3$, we finish our proof and conclude that $v \in M$, $C_1 \cap C_2 \subseteq M$, and finally: $M = C_1 \cap C_2$.
\end{proof}

Using these prerequisites, we are finally able to formulate the algorithm:

\begin{algorithm}[t]
    \SetKwInput{KwInput}{Input}
    \SetKwInput{KwOutput}{Output}
    \DontPrintSemicolon 
    \SetAlgoLined
    \KwInput{$(G, k)$, instance of \textsc{$\free_\prec(\set{P_3, \overline{K_3}})$-Vertex Splitting}}
    \KwOutput{\texttt{true} if $G$ can be made a member of $\free_\prec(\set{P_3, \overline{K_3}})$ using at most $k$ vertex splits, \texttt{false} otherwise}
    \BlankLine
    \uIf{$\overline{K_3} \prec G$}{
        \Return \texttt{false}\;
    }
    \uElseIf{$P_3 \not\prec G$}{
        \Return \texttt{true}\; \label{line:p3_k3_compl_else}
    }
    $M \gets \set{ v \in V(G) \mid \exists f \in \emb_\prec(P_3, G) \colon v \in f(V(P_3)) \land d_{P_3}(f^{-1}(v)) = 2}$\; \label{line:p3_k3_compl_m}
    \Return{$\overline{K_2} \not\prec G[M] \land |M| \leq k$} \label{line:p3_k3_compl_last}
    \caption{\textsc{$\free_\prec(\set{P_3, \overline{K_3}})$-Vertex Splitting}}
    \label{algorithm:splitting_polynomial_for_cluster_graphs_k3_complement}
\end{algorithm}

\begin{proposition}\label{lemma:splitting_polynomial_for_cluster_graphs_k3_complement}
   \textsc{$\free_\prec(\set{P_3, \overline{K_3}})$-Vertex Splitting} admits a polynomial-time algorithm. 
\end{proposition}
\begin{proof}
    We propose that \cref{algorithm:splitting_polynomial_for_cluster_graphs_k3_complement} is a polynomial-time algorithm for
    \textsc{$\free_\prec(\set{P_3, \overline{K_3}})$-Vertex Splitting}.

    \smallskip\emph{Correctness}:
    We claim that \texttt{true} is returned by \cref{algorithm:splitting_polynomial_for_cluster_graphs_k3_complement} when operating on an instance $(G, k)$ of \textsc{$\free_\prec(\set{P_3, \overline{K_3}})$-Vertex Splitting} if and only if $(G, k)$ is a positive instance of said problem.

    If $\overline{K_3} \prec G$, then $(G, k)$ is a negative instance,
    for $\overline{K_3}$ cannot be removed by way of vertex splitting (\cref{lemma:splits_maintain_cluster_graphs_with_tiny_clusters}).
    On the other hand, if we enter the else-branch on line \ref{line:p3_k3_compl_else},
    we know that $G \in \free_\prec{(\set{P_3, \overline{K_3}})}$, i.e., it already is a member of the desired class.

    Continuing onward, we may thus assume that $\overline{K_3} \not\prec G$ and $P_3 \prec G$.
    This already implies that $G$ is free of isolated vertices: An isolated vertex combined with the two endpoints of one embedded $P_3$ 
    would induce a $\overline{K_3}$ in $G$, a possibility already excluded. Furthermore, $K_3 \prec G$ gives that $G$ is non-empty and that $M$ constructed on line \ref{line:p3_k3_compl_m} will be non-empty too.

    We will now prove the correctness of the last statement on line \ref{line:p3_k3_compl_last}, i.e., \texttt{true} is returned if and only of $(G, k)$ is a positive instance.
    To prove this equivalence, we begin with the forwards direction, that is, if \texttt{true} is returned, then $(G, k)$ is a positive instance:
    Suppose $\overline{K_2} \not\prec G[M]$, meaning $G[M]$ is a clique, and $|M| \leq k$.
    We apply \cref{lemma:p3_midpoints_induce_scc} and obtain a sigma clique cover
    $\mathcal{C}$ of $G$ with $\wgt(\mathcal{C}) = |V(G)| + |M|$ and $|\mathcal{C}| \leq 2$.
    Using \cref{lemma:cvs_scc_reduction} and the fact that $G$ admits a sigma clique cover of weight $|V(G)| + k$,
    we conclude that 
    $(G,k)$ is a positive instance of \textsc{CVS}.
    We observe that the constructive proof of \cref{lemma:cvs_scc_reduction} never changes the cardinality
    of any sigma clique covers involved.
    Thus, more strongly, we know that
    any certificate (i.e., a splitting sequence) obtainable for \textsc{CVS} via \cref{lemma:cvs_scc_reduction} and $\mathcal{C}$
    must end in a cluster graph of at most two clusters.
    This implies it is not only free of $P_3$, but also free of $\overline{K_3}$.
    Thus, any certificate obtained in said manner also serves as a certificate of the
    more restricted \textsc{$\free_\prec(\set{P_3, \overline{K_3}})$-Vertex Splitting} problem
    considered here.

We will now prove the converse, that is, if $(G, k)$ is a positive instance,
then indeed \texttt{true} will be returned on line \ref{line:p3_k3_compl_last}:
Suppose $(G, k)$ is a positive instance of
\textsc{$\free_\prec(\set{P_3, \overline{K_3}})$-Vertex Splitting};
let $G_0, \dots, G_\ell$ with $G_0 = G$ be a splitting sequence certifying that fact.
Since $\free_\prec(\set{P_3, \overline{K_3}}) \subseteq \free_\prec(\set{P_3})$,
$G_0, \dots, G_\ell$ also serves as certificate for the instance $(G, k)$ of \textsc{CVS}.
Using \cref{lemma:cvs_scc_reduction}, we can build a sigma clique cover $\mathcal{C}$ of $G$ with $\wgt(\mathcal{C}) \leq |V(G)| + k$.
As already observed previously, the constructive proof of \cref{lemma:cvs_scc_reduction} never changes the cardinality
of any sigma clique covers involved. Thus, since $G_\ell$ consists of at most two connected components, $|\mathcal{C}| \leq 2$ follows.
Since $G$ is non-empty, we have that $|\mathcal{C}| \neq 0$.
If $|\mathcal{C}| = 1$, then $G$ is a clique, contradicting our assumption that $P_3 \prec G$.
Thus, we have $\mathcal{C} = \set{C_1, C_2}$ such that $C_1 \neq C_2$. If $C_1 \subseteq C_2$ or $C_2 \subseteq C_1$, then $G$ is again a clique, yielding the same contradiction as in the previous case.
Hence, we are allowed to apply \cref{lemma:scc_induces_p3_midpoints}
and conclude that $M = C_1 \cap C_2$.
Since the clique property is closed under intersection,
we conclude that $G[M]$ is a clique, implying $\overline{K_2} \not\prec G$.

It is straightforward to observe
that in order to destroy an embedding of a $P_3$ via vertex splitting, its embedded midpoint has to be split, and all vertices, where midpoints of some $P_3$
are embedded other than the vertex that is split, remain as such.
Hence, were $|M|$ to exceed $k$, no splitting sequence of length at most $k$ could
ever produce a graph free of $P_3$, which is a prerequisite for any positive instance.
Thus, we obtain $|M| \leq k$. We have derived both necessary conditions for \texttt{true} to be returned on line \ref{line:p3_k3_compl_last}. Hence, the proof is complete.

    \smallskip\emph{Running time}:
    We observe that all relevant operations performed in the algorithm reduce to enumerating, $\emb_\prec(H, G)$ where $H$ is fixed.
    This can be done in time polynomial with respect to $|V(G)|$ using the procedure
    employed in \cref{lemma:k_is_zero_polynomial}.
\end{proof}

\subsection{A Short Excursion to Ramsey Theory}
\label{section:ramsey}

Ramsey's Theorem is a fundamental result in the field of combinatorics that implies the following:
For every natural number $n \in \mathbb{N}^+$, there exists an integer $R(n)$ with the property that all graphs with at least $R(n)$ vertices must contain either a complete subgraph of $n$ vertices or an independent set of $n$ vertices \cite{ramsey}.

For our purposes, this means that the set of positive instances of 
\textsc{$\free_\prec(K_n, \overline{K_n})$-Vertex Splitting} forms a finite set for every $n \in \mathbb{N}^+$.

\begin{proposition}\label{lemma:k3_k3_compl_polynomial}
    Let $\mathcal{F}$ be a set of graphs with $\set{K_n, \overline{K_n}} \subseteq \mathcal{F}$ where $n \in \mathbb{N}^+$.
    Then, there is a polynomial-time algorithm for \textsc{$\free_\prec(\mathcal{F})$-Vertex Splitting}.
\end{proposition}
\begin{proof}
    By invoking Ramsey's Theorem,
    we conclude that for any graph $G$ of at least $R(n)$ vertices, we have $K_n \prec G$ or $\overline{K_n} \prec G$, that is, $G \not\in \free_\prec(\mathcal{F})$.

    Let $G$ be a graph and $k \in \mathbb{N}$. Consider the families of graphs $\mathcal{G}_0, \mathcal{G}_1, \dots$
    where $\mathcal{G}_i$ is the set of graphs that can be obtained from $G$ by performing at most $i$ vertex splits.
    Clearly, $(G, k)$ is a positive instance of \textsc{$\free_\prec(\mathcal{F})$-Vertex Splitting}
    if and only if $\mathcal{G}_{k} \cap \free_\prec(\mathcal{F}) \neq \varnothing$.
    Observe that $\mathcal{G}_k \cap \free_\prec(\mathcal{F}) \subseteq  \mathcal{G}_{R(n)} \cap \free_\prec(\mathcal{F})$,
    as all graphs that were split at least $R(n)$ times have at least $R(n)$ vertices and are hence not in $\free_\prec(\mathcal{F})$.

    With this, we can formulate an algorithm:
    Let $(G, k)$ be an instance of \textsc{$\free_\prec(\mathcal{F})$-Vertex Splitting}.
    If $\N{G} \ge R(n)$, we can safely reject the instance.
    Otherwise, we find that $G$ can only be isomorphic to one of finitely many graphs
    and can recognize the case applicable to $G$ in polynomial-time (\cref{lemma:k_is_zero_polynomial}).
    For each case, we can precompute $\mathcal{G}_{R(n)}$.
    To decide the instance $(G, k)$, it now suffices to check whether there is $G' \in \mathcal{G}_{R(n)}$, such that $G' \in \free_\prec(\mathcal{F})$.
    This can be done in polynomial-time by \cref{lemma:k_is_zero_polynomial}.
\end{proof}

\subsection{Forbidden Induced Subgraphs of at Most Three Vertices}
\label{section:dichotomy}

\begin{table}[t]
    \caption[Case distinction performed in the proof of \cref{theorem:leq3_complexity}.]{
        All cases for the set $\mathcal{F}$ of forbidden induced subgraphs of exactly three vertices used in the proof of \cref{theorem:leq3_complexity}.
        For each case, we list the complexity class \textsc{$\free_\prec(\mathcal{F})$-Vertex Splitting} falls into, and also reference the applicable proof.
    }
    \label{table:leq3}
    \begin{center}
        \begin{tabular}{@{}ccc@{}} \toprule
            $\mathcal{F}$ & Complexity & Proof \\ \midrule
            $\set{}$ & \P & \cref{lemma:splitting_polynomial_for_cluster_graphs_at_most_two} \\
            $\set{K_3}$ & \NP-Complete & \cref{theorem:single_biconnected_np_complete} \\
            $\set{\overline{K_3}}$ & \P & \cref{lemma:splitting_polynomial_for_cluster_graphs_at_most_two} \\
            $\set{P_3}$ & \NP-Complete & \cite{firbas_cluster_2023} \\
            $\set{\overline{P_3}}$ & \P & \cref{lemma:splitting_polynomial_for_cluster_graphs_at_most_two} \\
            $\set{K_3, \overline{K_3}}$ & \P & \cref{lemma:k3_k3_compl_polynomial} \\
            $\set{K_3, P_3}$ & \P & \cref{lemma:p3_k3_splitting_polynomial} \\
            $\set{\overline{K_3}, P_3}$ & \P & \cref{lemma:splitting_polynomial_for_cluster_graphs_k3_complement} \\
            $\set{K_3, \overline{P_3}}$ & \P & \cref{lemma:p3_complement_introduction} \\
            $\set{\overline{K_3}, \overline{P_3}}$ & \P & \cref{lemma:splitting_polynomial_for_cluster_graphs_at_most_two} \\
            $\set{P_3, \overline{P_3}}$ & \P & \cref{lemma:p3_complement_introduction} \\
            $\set{K_3, \overline{K_3}, P_3}$ & \P & \cref{lemma:k3_k3_compl_polynomial} \\
            $\set{K_3, \overline{K_3}, \overline{P_3}}$ & \P & \cref{lemma:k3_k3_compl_polynomial} \\
            $\set{K_3, P_3, \overline{P_3}}$ & \P & \cref{lemma:p3_complement_introduction} \\
            $\set{\overline{K_3}, P_3, \overline{P_3}}$ & \P & \cref{lemma:p3_complement_introduction} \\
            $\set{K_3, \overline{K_3}, P_3, \overline{P_3}}$ & \P & \cref{lemma:k3_k3_compl_polynomial} \\
            \bottomrule
        \end{tabular}
    \end{center}
\end{table}

In this very brief section, we can finally integrate the results of the preceding sections
and obtain the dichotomy we have worked towards:

\theoremleqthreecomplexity*
\begin{proof}
    It holds that $\mathcal{F} \subseteq \{K_0, K_1, K_2, \overline{K_2}, K_3, \overline{K_3}, P_3, \overline{P_3}\}$,
    for these are all graphs constructable using at most three vertices.
    If $\mathcal{F} \cap \{K_0, K_1, K_2, \overline{K_2}\} \neq \varnothing$,
    our problem admits a polynomial-time algorithm by \cref{lemma:splitting_polynomial_for_leq2}.
    Otherwise, refer to \cref{table:leq3} for a case distinction covering all remaining subsets not addressed by the former case.
\end{proof}

\subsection{\textsc{Split-} and \textsc{Threshold-Vertex Splitting}}
\label{section:split_thres}

In \cref{section:indestructible_fisgs},
we observed that vertex splitting is not suitable for transforming a graph into certain classes due to ``indestructible'' forbidden induced subgraphs.
In this section, we learn that this phenomenon not only concerns classes arising when systematically enumerating properties characterized by forbidden induced subgraphs of at most three vertices,
but rather also applies to well-known graph classes, such as split and threshold graphs.

Split and threshold graphs are two well-studied hereditary graph classes that are both characterizable via a set of forbidden induced subgraphs:
The class of \emph{threshold graphs} is given by $\free_\prec(\set{P_4, C_4, \overline{C_4}})$ \cite{trivially_perfect}, while the class of \emph{split graphs} is given by $\free_\prec(\set{C_4, C_5, \overline{C_4}})$ \cite{split_graphs}.
While both \textsc{Threshold-Node Deletion} and \textsc{Split-Node Deletion} are \NP-complete \cite{node_deletion},
the picture differs for the vertex splitting problem:

\begin{theorem}\label{theorem:threshold_split_poly}
\textsc{Threshold-} and \textsc{Split-Vertex Splitting} are in \P{}.
\end{theorem}
\begin{proof}
We proceed exactly as we did in \cref{lemma:p3_complement_introduction}. By \cref{lemma:splits_maintain_cluster_graphs_with_tiny_clusters}, we deduce that $\overline{C_4}$ cannot be destroyed via vertex splitting. However, it is easy to observe that the destruction of either $P_4$, $C_4$, or $C_5$ using a vertex split necessarily introduces a new $\overline{C_4}$. Thus, to decide an instance $(G, k)$ of \textsc{Threshold-Vertex Splitting} (resp. \textsc{Split-Vertex Splitting}), it suffices to determine whether $G$ is already a threshold graph (resp. a split graph). This can be done in polynomial time by \cref{lemma:k_is_zero_polynomial}.
\end{proof}

}                               %

\section{Biconnected Forbidden (Induced) Subgraphs and Beyond}
\label{section:biconnected}
\appendixsection{section:biconnected}

We now introduce a reduction framework and use it to show hardness for well-connected subgraphs.
The source problems of the reduction are a special family of vertex cover problems (\cref{section:special_vertex_cover}).
In \cref{section:a_construction}, we formulate the central construction on which our reductions are based.
In \cref{section:generic_reduction_technique}, we use said construction to introduce a generic reduction technique that can be applied to a large class of hereditary properties.
However, its correctness will depend on finding so-called \emph{admissible splitting configurations}.
In \cref{section:biconnected_fisgs}, we deal with finding such admissible splitting configurations for properties characterized by biconnected forbidden (induced) subgraphs.
Then, in \cref{section:higher_connected_fisgs}, we progress to higher levels of connectivity.

\subsection{A Special Flavor of Vertex Cover}
\label{section:special_vertex_cover}
For each fixed $\ell \in \mathbb{N}$, consider the \textsc{$2\ell$-Subdivided Cubic Vertex Cover} problem:

\begin{problem}[framed]{\textsc{$2\ell$-Subdivided Cubic Vertex Cover}}
    Input: & A tuple $(G^*, k)$, where $G^*$ is a $2\ell$-subdivision of a cubic graph $G$ and $k \in \mathbb{N}$. \\
    Question: & Is there a vertex cover $C$ of $G^*$ with $|C| \le k$?
\end{problem}
The \NP-hardness of this problem for each $\ell \in \mathbb{N}$ follows easily from a result by Uehara~\cite{vertex_cover} and ``folklore'' techniques.
\ifshort
For self-containedness we provide a formal proof in \cref{app:section:biconnected}.
\else
Nevertheless, to ensure comprehensiveness, we provide a formal proof.
\fi
\toappendix{
To show that \textsc{$2\ell$-Subdivided Cubic Vertex Cover} is \NP-hard for each $\ell \in \mathbb{N}$, we first examine the \textsc{Cubic Vertex Cover} problem.

\begin{problem}[framed]{\textsc{Cubic Vertex Cover}}
    Input: & A tuple $(G, k)$, where $G$ is a cubic graph and $k \in \mathbb{N}$. \\
    Question: & Is there a vertex cover $C$ of $G$ with $|C| \le k$?
\end{problem}

Uehara \cite{vertex_cover} showed that \textsc{Cubic Vertex Cover}, with the additional constraint
that the input graph is 3-connected, planar, and of girth at least four, is \NP-hard.
Since the set of problem instances for \textsc{Cubic Vertex Cover} is a superset of the problem instances of this strengthened variant,
we immediately know that \textsc{Cubic Vertex Cover} is \NP-hard too.

\begin{problem}[framed]{\textsc{$2\ell$-Subdivided Cubic Vertex Cover}}
    Input: & A tuple $(G^*, k)$, where $G^*$ is a $2\ell$-subdivision of a cubic graph $G$ and $k \in \mathbb{N}$. \\
    Question: & Is there a vertex cover $C$ of $G^*$ with $|C| \le k$?
\end{problem}

\begin{figure}
    \begin{center}
        \includegraphics[]{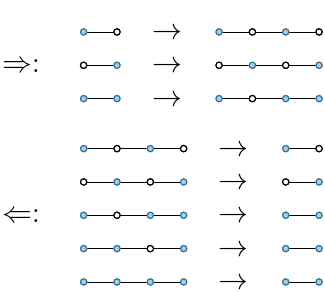}
    \end{center}
    \caption[Illustration of the replacements performed in the proof of \cref{lemma:subdivided_vc_reduction}.]{Illustration of all possible replacements performed in both directions of the proof of \cref{lemma:subdivided_vc_reduction}.
            The coloring indicates which vertices are part of the respective vertex cover.}
    \label{figure:subdivided_vc}
\end{figure}

The following lemma gives a straightforward reduction from \textsc{Cubic Vertex Cover} to \textsc{$2\ell$-Subdivided Cubic Vertex Cover} for any fixed $\ell$.
In the forward direction, we subdivide a graph that comes with a vertex cover step by step, each time replacing an edge with a path of four vertices.
A vertex cover for the new graph can then be obtained by including one of the two subdivision vertices into the vertex cover.
In the opposite direction, we convert a subdivided graph that comes with a vertex cover into its contracted equivalent by substituting paths consisting of four vertices with individual edges. We observe that, when removing vertices that underwent contraction, the vertex cover of the subdivided graph transforms into a reduced-size vertex cover of the contracted graph.
Reference \cref{figure:subdivided_vc} for an illustration of both transformations.

\begin{lemma}\label{lemma:subdivided_vc_reduction}
    Let $G$ be a cubic graph, $\ell,k \in \mathbb{N}$
    and $G^*$ be a $2\ell$-subdivision of $G$.
    Then, $(G, k)$ is a positive instance of
    \textsc{Cubic Vertex Cover} if and only if
    $(G^*, k + \ell \M{G})$ is a positive instance
    of \textsc{$2\ell$-Subdivided Cubic Vertex Cover}.
\end{lemma}
\begin{proof}
    $(\Rightarrow)\colon$
    Let $C$ be a vertex cover of $G$ with $|C| \leq k$.
    Consider $v_1v_4 \in E(G)$.
    Without loss of generality, we have $v_1 \in C$.
    Create $G'$ by subdividing $v_1v_4$ into the path $v_1v_2v_3v_4$
    and set $C' \coloneqq C \cup \set{v_3}$.
    We observe that $C'$ is a vertex cover of $G'$.
    Now, starting with $G$, by $\ell$ such operations per edge of $G$, totaling $\ell \M{G}$ operations,
    we can subdivide all edges of $G$ to construct $G^*$ with an accompanying vertex cover $C^*$ where
    \begin{equation*}
        |C^*| = |C| + \ell \M{G} \leq k + \ell \M{G}.
    \end{equation*}

    $(\Leftarrow)\colon$
    Let $C^*$ be a vertex cover of $G^*$ with $|C^*| \leq k + \ell \M{G}$.
    Consider a path $v_1v_2v_3v_4$ in $G^*$ where $\deg_{G^*}(v_2) = \deg_{G^*}(v_3) = 2$.
    Create ${G^{*}}'$ by replacing this path by the single edge $v_1v_4$ and set ${C^{*}}' \coloneqq C^* \setminus \set{v_2, v_3}$.
    Observe that $|{C^*}'| \leq |C^*| - 1$ and that ${C^*}'$ is a vertex cover of ${G^*}'$.
    Starting with $G^*$, by $\ell \M{G}$ operations as just described,
    we can construct $G$ by ``undoing'' the subdivisions, and additionally obtain an accompanying vertex cover $C$ where
    \begin{equation*}
        |C| \leq |C^*| - \ell \M{G} \leq k.\qedhere
    \end{equation*}
\end{proof}
Using this reduction the \NP-hardness proof is immediate:
\begin{lemma}\label{lemma:two_l_subdivided_cubic_vc_np_hard}
    \textsc{$2\ell$-Subdivided Cubic Vertex Cover} is \NP-hard for each $\ell \in \mathbb{N}$.
\end{lemma}
\begin{proof}
    By \cref{lemma:subdivided_vc_reduction}, for each $\ell \in \mathcal{N}$ we have
    \begin{equation*}
        \text{\textsc{Cubic Vertex Cover}} \leq_{\P} \text{\textsc{$2\ell$-Subdivided Cubic Vertex Cover}}.
    \end{equation*}
    Furthermore, \textsc{Cubic Vertex Cover} is \NP-hard.
    Therefore, \textsc{$2\ell$-Subdivided Cubic Vertex Cover} is \NP-hard for all $\ell \in \mathbb{N}$.
\end{proof}
}

\subsection{Splitting Configurations and the Central Construction}
\label{section:a_construction}
We start by introducing some notation for describing a particular split and subsequently formalize the notion of a \emph{splitting configuration} that encodes a strategy of how vertices of a given graph are to be split.
\begin{definition}\label{definition:split_shorthand}
    Let $H$ be a graph, $v \in V(H)$,
    $X_1, X_2 \subseteq N_{H}(v)$ with $X_1 \cup X_2 = N_{H}(v)$,
    and $v_1, v_2$ two distinct vertices.
    Further, let $H'$ be the graph obtained by splitting 
    $v$ into $v_1$ and $v_2$ while setting $N_{H'}(v_1) = X_1$, 
    $N_{H'}(v_2) = X_2$.
    Then, we identify $H'$ with the shorthand $\Split(H, v, X_1, X_2, v_1, v_2)$.
\end{definition}

Next, we introduce the concept of a \emph{splitting configuration}.
Intuitively, a splitting configuration consists of a graph $H$,
a selection of two of its vertices ($a$ and $b$),
and an encoding of a specific strategy of how to split $a$ and $b$ in $H$.

\begin{definition}\label{definition:splitting_configuration}
    Let $H$ be a graph,
    $a, b \in V(H)$ distinct vertices,
    $A_1, A_2 \subseteq N_H(a)$, and $B_1, B_2 \subseteq N_H(b)$, such that
    $A_1 \cup A_2 = N_H(a)$, $B_1 \cup B_2 = N_H(b)$, and
    $\set{A_1, A_2, B_1, B_2} \cap \set{\varnothing} = \varnothing$.
    Then, $(H, a, A_1, A_2, b, B_1, B_2)$ is called a
    \emph{splitting configuration}.
    If additionally $A_1 \cap A_2 = B_1 \cap B_2 = \varnothing$, we speak 
    of a \emph{disjoint splitting configuration}.
    Furthermore, we say the splitting configuration is \emph{based upon} $\mathcal{F}$ if $H \in \mathcal{F}$.
\end{definition}

A splitting configuration $C$ serves as the ``atomic'' building block of our construction.
In addition to the splitting configuration, we consider a directed ``skeleton'' graph $\vec{G}$ and a subset $S$ of $\vec{G}$'s vertices.
In total, these three values will determine the graph $\Constr(\vec{G}, C, S)$.

\looseness=-1
Our reductions maps an instance $(G, k)$ of $2\ell$-\textsc{Subdivided Cubic Vertex Cover} to some vertex splitting problem (depending on the graph property in question) and an instance $(\Constr(\vec{G}, C, \varnothing), k)$, where $\vec{G}$ is an orientation of $G$ and $C$ is some suitable splitting configuration.
Thus, we always set the graph $\vec{G}$ to be an orientation of an instance of $2\ell$-\textsc{Subdivided Cubic Vertex Cover} for some $\ell \in \mathbb{N}$;
note that when computing the reduction itself, we simply set $S = \varnothing$.
However, we will set $S$ to non-empty sets when performing the forward direction of the correctness proof.
There, we need to find splitting sequences starting with $\Constr(\vec{G}, C, \varnothing)$ and ending with a graph that is free of forbidden (induced) subgraphs.
Then, by subsequently introducing elements to $S$ that stem from a given vertex cover of $G$,
we can use the $\Constr(\vec{G}, C, S)$ notation to directly construct each member of the sequence. 

\looseness=-1
Informally, building the graph $\Constr(\vec{G}, C, S)$ amounts to the following sequence of steps:
Replace each arc of $\vec{G}$ with a copy of $H$. Here, $a$ and $b$ act as \emph{attachment points} or \emph{ends}, and the orientation of each arc dictates whether $H$ shall be inserted ``forwards'' or ``backwards''.
We often call such copies of $H$ \emph{edge gadget}.
If $S$ is empty, the construction is complete.
Otherwise, for each $s \in S$ we perform a vertex split of the corresponding attachment point.
The number of affected copies of $H$ equals the degree of the vertex corresponding to the attachment point in $\vec{G}$.
In this split, each individual copy of $H$ attached to the attachment point is split according to the splitting configuration $C$.
Consult \cref{figure:constr_exampl} for an example with concrete values.

\begin{figure}[t]
    \begin{center}
        \includegraphics{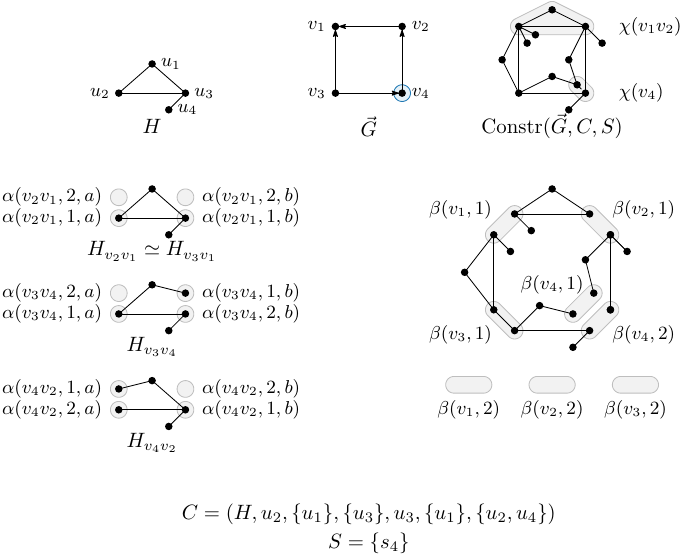}
    \end{center}
    \caption[Example of \cref{definition:constr}.]{Example of \cref{definition:constr}.
    The construction is carried out for the ``skeleton'' graph $\vec{G}$, a splitting configuration $C$, and the set of vertices $S$ marked in blue.
    The edge gadget graph is $H$;
    its ``$a$-end'' is $u_2$ and its ``$b$-end'' is $u_3$.
    The subscript of $\chi$, $\Constr(\vec{G}, C, S)$, is dropped for brevity.}
    \label{figure:constr_exampl}
\end{figure}

Below, whenever we encounter a graph $G'$ that is a copy of a graph $G$, we use $v_{G'}$ to denote the vertex that corresponds to $v \in V(G)$ in $G'$.
We also do likewise for sets of vertices.

\subparagraph{Towards defining \boldmath $\Constr(\vec{G}, C, S)$}
    Let $\vec{G}$ be a directed, oriented graph without loops,
    $C$ a splitting configuration with $C = (H, a, A_1, A_2, b, B_1, B_2)$,
    and $S \subseteq V(G)$.
    We aim to define the graph
    $\Constr(\vec{G}, C, S)$ and the map $\chi_{\Constr(\vec{G}, C, S)}$.
    For this we first need further notation for the gadget graphs $H_e$, and two maps $\alpha$ and $\beta$ that specify attachment points.

    With each arc $e = v_av_b \in E(\vec{G})$, we associate a fresh copy of $H$ and call it $H'_{e}$.
    In correspondence with the notational convention described above, the vertices $a_{H'_{e}}, b_{H'_{e}}$ and the sets of vertices ${A_1}_{H'_{e}}, {A_2}_{H'_{e}}, {B_1}_{H'_{e}}, {B_2}_{H'_{e}}$
    denote the corresponding vertex or set of vertices of $H$ in its copy, $H'_{e}$.
    We obtain $H_e$ by splitting a subset of $\set{a_{H'_{e}}, b_{H'_{e}}}$ in $H'_{e}$.
    Whether we split zero, one, or two vertices is dictated by $S$ ($a_{H'_{e}}$ is split if{}f $v_a \in S$, $b_{H'_{e}}$ is split if{}f $v_b \in S$);
    the precise manner vertices are split is dictated by the splitting configuration $C$. More specifically, the neighborhoods of the descendant vertices of $a_{H'_{e}}$ (resp. $b_{H'_{e}}$) are given by ${A_1}_{H'_{e}}, {A_2}_{H'_{e}}$ (resp. ${B_1}_{H'_{e}}, {B_2}_{H'_{e}}$).
    With this, we can specify formally how $H_e$ is obtained from each $e = v_av_b$ of $E(\vec{G})$:
    \begin{equation*}
        H_{e} \coloneqq
        \begin{cases}
            H'_{e} & \text{if } v_a \not\in S, v_b \not\in S,\\ 
            \Split(H'_{e}, a_{H'_{e}}, {A_1}_{H'_e}, {A_2}_{H'_e}, {a_1}_{H_e}, {a_2}_{H_e} ) & \text{if } v_a \in S, v_b \not\in S,\\ 
            \Split(H'_{e}, b_{H'_{e}}, {B_1}_{H'_e}, {B_2}_{H'_e}, {b_1}_{H_e}, {b_2}_{H_e}) & \text{if } v_a \not\in S, v_b \in S \text{, and}\\
            \Split(\Split(H'_{e}, a_{H'_{e}}, {A_1}_{H'_e}, {A_2}_{H'_e}, {a_1}_{H_e}, {a_2}_{H_e}),\\
            \hphantom{\Split(\Split(H'_{e},\hspace{0.08cm}} b_{H'_{e}}, B_1^*, B_2^*, {b_1}_{H_e}, {b_2}_{H_e}) &  \text{otherwise,}\\
        \end{cases}
    \end{equation*}
    where $B_1^*$ (resp. $B_2^*$) denote the descendant vertices of ${B_1}_{H'_e}$ (resp. ${B_2}_{H'_e}$) with respect to the split described by $\Split(H'_{e}, a_{H'_{e}}, {A_1}_{H'_e}, {A_2}_{H'_e}, {a_1}_{H_e}, {a_2}_{H_e})$.\footnote{This additional care is required to cover the case when $a$ and $b$ are neighbors in $H$.}
    The set $\set{H_e \mid e \in E(\vec{G})}$ of edge gadgets provides the basic building blocks of $\Constr(\vec{G}, C, S)$.
    Note that the vertex sets of all $H_e$ with $e \in E(\vec{G})$ are disjoint;
    to construct the final graph $\Constr(\vec{G}, C, S)$, we need to join the edge gadgets according to the structure of $\vec{G}$.
    
    For this purpose, we designate two numbered \emph{attachment points} for the $a$-end, and two numbered attachment points for the $b$-end of each $H_e$,
    where an attachment point is a possibly empty subset of $H_e$'s vertices.
    Consider some edge gadget $H_e$ and one of its ends, say, the $a$-end: If $a_{H'_e}$ was split when building $H_e$,
    the first attachment point of $H_e$'s $a$-end consists of a singleton set containing the first descendant of $a_{H'_e}$,
    while the second attachment point of $H_e$'s $a$-end consists of a singleton set containing the second descendant of $a_{H'_e}$.
    If otherwise $a_{H'_e}$ was not split when building $H_e$, then the first attachment point of $H_e$'s $a$-end is the set $\set{a_{H_e}}$, and the second attachment point of $H_e$'s $a$-end is the empty set.
    Later, when building $\Constr(\vec{G}, C, S)$, we will select attachment points from different edge gadgets. Then, we will merge all the vertices of the selected attachment points into a single vertex for each such selection of attachment points.

    Formally, we determine the attachment points using the map $\alpha(\cdot, \cdot, \cdot)$.
    We use $e \in E(\vec{G})$ to select an edge gadget, $x \in \set{a, b}$ to select the end of $H_e$, and $i \in \set{1, 2}$ to select either the first or second attachment point of this end of $H_e$.
    Then, we define
    \begin{equation*}
        \alpha(e, i, x) \coloneqq
        \begin{cases}
            \set{ {x_i}_{H_e} } & \text{if } v_x \in S,\\
            \set{ x_{H_e} } & \text{if } v_x \not\in S \land i = 1 \text{, and}\\
            \varnothing & \text{if } v_x \not\in S \land i = 2.\\
        \end{cases}
    \end{equation*}

    \looseness=-1
    Above we have defined the building blocks (edge gadgets) to join, as well as attachment points describing at which vertices edge gadgets can be joined.
    It remains to incorporate the edge orientations and structure of $\vec{G}$ to determine how exactly the edge gadgets are assembled into $\Constr(\vec{G}, C, S)$.
    Intuitively, we replace each arc of $\vec{G}$ with the corresponding edge gadget, and use the orientation of $\vec{G}$ to determine which way the edge gadget is to be inserted, since each edge gadgets has two ``ends''.
    Below we list a series of equivalence classes, that is, sets of vertices.
    The final graph $\Constr(\vec{G}, C, S)$ is then built by composing all
    $(H_e)_{e \in E(\vec{G})}$ into a single graph and merging all equivalent vertices into one representative vertex each.

    \looseness=-1
    Concretely, for each $v \in V(\vec{G})$, we define two equivalence classes, stemming from the circumstance that we have two attachment points per edge gadget end.
    The set of equivalence classes is given by
        $\bigcup_{v \in V(\vec{G})} \set{\beta(v, 1), \beta(v, 2)}$,
    where for each $v \in V(\vec{G})$ and $i \in \set{1, 2}$, we define
    \begin{align*}
        \beta(v, i) \coloneqq \set{
            \parens{\bigcup_{u \in N^-_{\vec{G}}(v)} \alpha(uv, i, b)} \cup \parens{\bigcup_{u \in N^+_{\vec{G}}(v)} \alpha(vu, i, a)}
        }.
    \end{align*}
    Intuitively, $\beta$ does the following:
    Consider $v \in V(\vec{G})$.
    For all incoming arcs, we compute the union of the first (resp. the second) attachment points of the ``$b$-end'' of the edge gadgets corresponding to the incoming arcs,
    and likewise, for all outgoing arcs, we compute the union of the first (resp. the second) attachment points of the ``$a$-end'' of the edge gadgets corresponding to the outgoing arcs.

    With this, all equivalence classes are fully defined.
    Remember that each equivalence class is of the form $\beta(v, 1)$ or $\beta(v, 2)$ where $v \in V(\vec{G})$.
    Thus we can specify the main construction.%
    \begin{definition}\label{definition:constr}
      The graph $\Constr(\vec{G}, C, S)$ is constructed by composing all $(H_e)_{e \in E(\vec{G})}$ into a single graph and merging all equivalent vertices, as defined above, into one representative vertex each.
    \end{definition}
    
    Later, when reasoning about the construction, we will need a way to select the vertices that stem from either a single edge gadget, or from the intersection of multiple edge gadgets. 
    Note that the step where we joined the copies of $H$ defines a function $f$ that assigns to each vertex of an $H$-copy the vertex of $\Constr(\vec{G}, C, S)$ it corresponds to.
    Using $f$, we formulate the auxiliary function 
        $\chi_{\Constr(\vec{G}, C, S)} \colon V(\vec{G}) \cup E(\vec{G}) \to \mathcal{P}(V(\Constr(\vec{G}, C, S)))$
    as follows.
    We map each arc $e \in E(\vec{G})$ to the set of vertices
    that correspond to the vertices of $H_e$ in $\Constr(\vec{G}, C, S)$, that is,
        $\chi(e)_{\Constr(\vec{G}, C, S)} \coloneqq f( V(H_e) )$,
    and each vertex $v \in V(\vec{G})$ to either a set of a single vertex of $\Constr(\vec{G}, C, S)$ if
    $v \not\in S$, or two distinct vertices otherwise, that is,
        $\chi(v)_{\Constr(\vec{G}, C, S)} \coloneqq f(\beta(v, 1) \cup \beta(v, 2))$.
    These vertices ``sit at the intersection'' of different $H$-copies; we will call these vertices either simply \emph{ends}, or more specifically \emph{$a$- or $b$-ends} (of an edge gadget), respectively.
    Reference \cref{figure:constr_exampl} for an example.

Abstracting from a single instantiation of our construction,
we also introduce notation to capture the class of all possible constructions based on a given splitting configuration and an undirected graph together with all of its vertex covers.
\begin{definition}\label{definition:all_constr}
    Let $G$ be a simple graph
    and $C$ a splitting configuration.
    Then, we write $\AllConstr(G, C)$ to describe the set of all
    $\Constr(\vec{G}, C, S)$, where $\vec{G}$ is an orientation of $G$
    and $S \subseteq V(G)$ is a vertex cover of $G$.
\end{definition}

\subsection[A Generic Reduction Technique]{A Generic Reduction Technique and Admissible Splitting Configurations}
\label{section:generic_reduction_technique}
\looseness=-1
Consider a hereditary property characterized by a set of forbidden (induced) subgraphs $\mathcal{F}$.
We devise a general method to show that $\free_{\prec/\subseteq}(\mathcal{F})$-\textsc{Vertex Splitting} is \NP-complete
by calculating $\ell \in \mathbb{N}$ depending on $\mathcal{F}$ and reducing
\textsc{$2\ell$-Subdivided Cubic Vertex Cover} to $\free_{\prec/\subseteq}(\mathcal{F})$-\textsc{Vertex Splitting}.
The instance for the vertex splitting problem will be given by $\Constr(\vec{G}, C, \varnothing)$, where $\vec{G}$ is an arbitrarily chosen orientation of the given $2\ell$-subdivided graph $G$, and $C$ is a splitting configuration based on $\mathcal{F}$.
The correctness of this reduction depends on the choice of $C$, meaning that the approach can fail.
In this subsection, we define the property of \emph{admissibility} for a splitting configuration and show that, when a splitting configuration is admissible, the reduction outlined above is correct.
The remaining sections then deal with finding admissible splitting configurations for various classes of hereditary properties.

The simpler part of the correctness proof is the backward direction, that is, extracting a vertex cover from a splitting sequence that destroys all forbidden subgraphs, which works independently of choice of $C$ or $\ell$:
\begin{restatable}[\appsymb]{lemma}{lemmasplittingsequencetovertexcover}\label{lemma:splitting_sequence_to_vertex_cover}
    Let $G$ be a graph,
    $\vec{G}$ an orientation of $G$,
    $\mathcal{F}$ a family of graphs,
    $C$ a splitting configuration based upon $\mathcal{F}$,
    and $G_0, \ldots, G_k$ a splitting sequence such that
    $G_0 = \Constr(\vec{G}, C, \varnothing)$ and
    $G_k \in \free_\subseteq(\mathcal{F})$ (resp. $G_k \in \free_\prec(\mathcal{F})$).
    Then, there exists a vertex cover of $G$ with size at most $k$.
\end{restatable}
\appendixproof{lemma:splitting_sequence_to_vertex_cover}{
\ifshort\lemmasplittingsequencetovertexcover*\fi
\begin{proof}
    Let $A \subseteq V(G_0)$ be the set of all ancestors
    of vertices that are split in the splitting sequence in $G_0$.
    Now, define a mapping $f \colon A \to V(G)$ as follows:
    For each $a \in A$, there either is exactly one $v \in V(G)$, such that $a \in \chi_{G_0}(v)$, or no such vertex.
    In the former case, we set $f(a) \coloneqq v$;
    in the latter, we find $e \in E(G)$ with $a \in \chi_{G_0}(e)$,
    arbitrarily choose one of its endpoints, and map $f(a)$ to it.
    We claim that $\range(f)$ is a vertex cover of $G$ of size at most $k$.

    Exactly $k$ vertices were split in the splitting sequence;
    some split vertices may share a common ancestor in $G_0$.
    Thus, $|A| \leq k$. From this and the observation that $|\range(f)| \leq |A|$,
    it follows that $\range(f) \leq k$.

    It remains to show that $\range(f)$ is a vertex cover of $G$:
    Let $uv \in E(G)$.
    We know that $G_0[\chi_{G_0}(uv)] \simeq F$ for some $F \in \mathcal{F}$ because $C$ is based upon $\mathcal{F}$.
    Additionally, we required that $G_k \in \free_\subseteq(\mathcal{F})$ (resp. $G_k \in \free_\prec(\mathcal{F})$).
    Thus, there is some $a \in A \cap \chi_{G_0}(uv)$.
    Since $f(a) \in \set{u, v}$, it follows that $f(a) \in \range(f)$ is a suitable witness for $uv$ being covered by $\range(f)$.
\end{proof}
}

\looseness=-1
We now tackle the more difficult part of the correctness proof, the forward direction, where we use a vertex cover to find a splitting sequence that destroys all forbidden subgraphs in the construction.
Here, the choice of $C$ and $\ell$ will matter.
We are given a vertex cover of the ``skeleton graph'' $\vec{G}$ and split all of the attachment points in the construction according to a corresponding splitting configuration.
In the final graph of the splitting sequence, the whole construction needs to be free of embeddings of forbidden (induced) subgraphs.
This can be rephrased as two separate properties that a splitting configuration must guarantee when applying our construction to any conceivable instance of \textsc{$2\ell$-Subdivided Cubic Vertex Cover} and splitting it according to a vertex cover:%
\begin{itemize}
\item There are no embeddings of forbidden (induced) subgraphs contained entirely within any individual edge gadget.
\item There are no embeddings of forbidden (induced) subgraphs reaching from one edge gadget to a neighboring edge gadget.
\end{itemize}%
\looseness=-1%
In \cref{definition:admissible}, we formalize both these requirements.
The first requirement is addressed through the concept of intra-edge embedding-free splitting configurations,
while the second requirement is formalized using the notion of separating splitting configurations.
Intuitively, using a separating splitting configuration has the property that, when carrying out $\Constr(G, C, S)$ where $S$ is a vertex cover of $G$, regardless of the orientation of $G$ and choice of $S$,
each embedding of a forbidden (induced) subgraph in $\Constr(G, C, S)$ is confined to a single edge gadget, provided that $G$ is an $L$-subdivision of a graph for large enough $L$.
Note that the requirement on $\mathcal{F}$ to be of bounded diameter will serve to guarantee that such an $L$ can be found. 
If a splitting configuration is both intra-edge embedding-free and separating, we say it is admissible.
\begin{definition}\label{definition:admissible}
    Let $\mathcal{F}$ be a family of graphs of bounded diameter
    with
    $H \in \mathcal{F}$ and $C = (H, a, A_1, A_2, b, B_1, B_2)$
    a splitting configuration.
    Furthermore, let $L \coloneqq 2 \cdot \max_{F \in \mathcal{F}} \diam(F)$.
    Then, $C$ is called \emph{separating}
    for $\mathcal{F}$ if for all graphs $G$ that are an $L$-subdivision of some cubic graph, we have
    \begin{equation*}
        \forall G^* \in \AllConstr(G, C) \colon \forall F \in \mathcal{F} \colon \forall \pi \in \emb_\subseteq(F, G^*) \colon \exists e \in E(G) \colon \range(\pi) \subseteq \chi_{G^*}(e).
    \end{equation*}
    Additionally, we say $C$ is \emph{intra-edge embedding-free} for $\mathcal{F}$ if $\AllConstr(K_2, C) \subseteq \free_\subseteq(F)$.
    Finally, the splitting configuration $C$ is called \emph{admissible} for $\mathcal{F}$ if it is both separating for $\mathcal{F}$ as well as intra-edge embedding-free for $\mathcal{F}$.
\end{definition}

Once we have obtained such an admissible splitting configuration, we can also tackle the opposite direction of the correctness proof of our reduction.
\begin{restatable}[\appsymb]{lemma}{lemmareductionadmissible}\label{lemma:reduction_admissible}
    Let $\mathcal{F}$ be a family of graphs of bounded diameter,
    $C$ a splitting configuration admissible for $\mathcal{F}$,
    $L \coloneqq 2 \cdot \max_{F \in \mathcal{F}} \diam(F)$,
    and $(G, k)$ an instance
    of \textsc{$L$-Subdivided Cubic Vertex Cover}.
    Then, $(G, k)$ is a positive instance of
    \textsc{$L$-Subdivided Cubic Vertex Cover}
    if and only if $(\Constr(\vec{G}, C, \varnothing), k)$
    is a positive instance of \textsc{$\free_\subseteq(\mathcal{F})$-VS} (resp. \textsc{$\free_\prec(\mathcal{F})$-VS})
    where $\vec{G}$ is an orientation of $G$.
\end{restatable}
\appendixproof{lemma:reduction_admissible}{
\ifshort\lemmareductionadmissible*\fi
\begin{proof}
    We only prove the variant of this statement involving \emph{induced} subgraphs; the proof for the alternative regarding subgraphs is analogous.

    $(\Rightarrow)\colon$
    Let $S$ be a vertex cover of $G$ with $|S| \le k$.
    Then, by fixing a total order of $S$, we can obtain a splitting sequence $G_0, \ldots, G_k$ with
    $G_0 = \Constr(\vec{G}, C, \varnothing)$ and
    $G_k = \Constr(\vec{G}, C, S)$.
    It remains to show that $G_k \in \free_\prec(\mathcal{F})$.
    Towards a contradiction, suppose the opposite, that is, there is $F \in \mathcal{F}$ with $\pi \in \emb_\prec(F, G_k)$.
    Since $C$ is separating for $\mathcal{F}$, we deduce that there is $e \in E(\vec{G})$ such that $\range(\pi) \subseteq \chi_{G_k}(e)$.
    But $G_k[\chi_{\vec{G}}(e)]$ is isomorphic to some graph of $\AllConstr(K_2, C)$ and by the admissibility of $C$ for $\mathcal{F}$, we have $\AllConstr(K_2, C) \subseteq \free_\subseteq(\mathcal{F}) \subseteq \free_\prec(\mathcal{F})$, a contradiction to the existence of $\pi$.
    Thus $G_k \in \free_\prec(\mathcal{F})$.

    $(\Leftarrow)\colon$ The claim follows directly from \cref{lemma:splitting_sequence_to_vertex_cover}.
\end{proof}
}

Having shown both directions of the correctness proof, one generically, the other only for admissible splitting configurations, it is easy to establish the following \NP-hardness result:
\begin{restatable}[\appsymb]{lemma}{lemmanpcompleteadmissible}\label{lemma:np_complete_admissible}
    Let $\mathcal{F}$ be a family of graphs
    of bounded diameter and let $C$ be a splitting configuration admissible for $\mathcal{F}$.
    Then, \textsc{$\free_\prec(\mathcal{F})$-VS} and \textsc{$\free_\subseteq(\mathcal{F})$-VS} are \NP-hard.
\end{restatable}
\appendixproof{lemma:np_complete_admissible}{
\begin{proof}
    Without loss of generality, we only consider \textsc{$\free_\prec(\mathcal{F})$-VS}.
    To establish \NP-hardness, we 
    construct a polynomial-time many-one reduction from \textsc{$L$-Subdivided Cubic Vertex Cover} to \textsc{$\free_\prec(\mathcal{F})$-VS}:
    Let $(G, k)$ be an instance of \textsc{$L$-Subdivided Cubic Vertex Cover}.
    Using \cref{lemma:reduction_admissible}, we obtain an instance of $\textsc{$\free_\prec(\mathcal{F})$-VS}$ in polynomial time 
    that is equivalent to the original instance $(G, k)$ of \textsc{$L$-Subdivided Cubic Vertex Cover}.
    Thus, since \textsc{$L$-Subdivided Cubic Vertex Cover} is \NP-hard, so is \textsc{$\free_\prec(\mathcal{F})$-VS}.
\end{proof}
}

\subsection{Biconnected Forbidden Subgraphs}
\label{section:biconnected_fisgs}

In \cref{lemma:np_complete_admissible}, we established a method for obtaining \NP-hardness results for vertex splitting problems,
provided an appropriate admissible splitting configuration existed.
This naturally prompts the question of how to find such splitting configurations.
This subsection address this question for the case of biconnected forbidden (induced) subgraphs.
Before proceeding, we need to introduce one last piece of notation: the \emph{width} of a splitting configuration, denoted by $\width(\cdot)$, represents the minimum distance between the two descendants of a split endpoint, $a$ and $b$, respectively, after $H$ has been split according to the splitting configuration.
\begin{definition}\label{definition:width_splitting_configuration}
    Let $C = (H, a, A_1, A_2, b, B_1, B_2)$ be a splitting configuration, graph
    $H_1 \coloneqq \Split(H, a, A_1, A_2, a_1, a_2)$, graph
    $H_2 \coloneqq \Split(H, b, B_1, B_2, b_1, b_2)$, where $a_1, a_2, b_1$ and $b_2$ are fresh vertices.
    Then, we define $\width(C) \coloneqq \min \set{d_{H_1}(a_1, a_2), d_{H_2}(b_1, b_2)}$.
\end{definition}

\looseness=-1
\iflong
  In the next two lemmas we show that, given a splitting configuration of a certain width that is not separating (for some $\mathcal{F}$), we can derive a new splitting configuration of increased width (\cref{definition:width_splitting_configuration}).
\else
  We now show that, given a splitting configuration of a certain width that is not separating (for some $\mathcal{F}$), we can derive a new splitting configuration of increased width (\cref{definition:width_splitting_configuration}).
\fi
Since we cannot apply this process ad infinitum (when restricted to $\mathcal{F}$ of bounded circumference), we will ultimately arrive at a separating splitting configuration (\cref{lemma:existence_splitting_config_for_biconnected}).

\begin{lemma}\label{lemma:splitting_config_of_greater_width}
    Let $\mathcal{F}$ be a family of biconnected graphs
    of bounded diameter
    and let $C^0 = (H^0, a^0, A_1^0, A_2^0, b^0, B_1^0, B_2^0)$ be a disjoint
    splitting configuration of finite width such that there is $H^0 \in \mathcal{F}$ that is not separating for $\mathcal{F}$.
    Then, there exists a disjoint splitting configuration
    $C^1 = (H^1, a^1, A_1^1, A_2^1, b^1, B_1^1, B_2^1)$ with $H^1 \in \mathcal{F}$
    of finite width satisfying $\width(C^1) > \width(C^0)$.
\end{lemma}
\begin{proof}
    As $C^0$ is not separating for $\mathcal{F}$, there is a 
    graph $G$ that is an $L$-subdivision of some cubic graph $G'$,
    $G^* \in \AllConstr(G, C^0)$,
    $F \in \mathcal{F}$ with $\pi \in  \emb_\subseteq(F, G^*)$, as well as 
    distinct $u, v, w \in V(G)$ with $L \coloneqq 2 \cdot \max_{F \in \mathcal{F}} \diam(F)$, 
    $\range(\pi) \cap (\chi_{G^*}(uv) \setminus \chi_{G^*}(v)) \neq \varnothing$, and
    $\range(\pi) \cap (\chi_{G^*}(vw) \setminus \chi_{G^*}(v)) \neq \varnothing$.

    In other words, $G^*$ is a graph constructed according to \cref{definition:constr}
    using a highly subdivided cubic graph ($G$) as basis, where its edges were replaced by some forbidden graph $H \in \mathcal{F}$, and was split at the ``attachment points'' of edge gadgets according to some vertex cover of $G'$ and the splitting configuration $C^0$.
    For this graph, we are provided a witness certifying that the splitting configuration $C^0$ is not separating with respect to $\mathcal{F}$ in the form of
    an embedding $\pi$ of $F \in \mathcal{F}$ into $G^*$, where the embedding of $F$ is not constrained to a single edge gadget, but rather uses at least vertices 
    of two neighboring edge gadgets (of edges $uv, vw \in E(G')$),  $\chi_{G^*}(uv)$ and $\chi_{G^*}(vw)$, but not those in the shared intersection $\chi_{G^*}(v)$; the embedding thus ``goes across'' two edge gadgets.
    Notice that $\pi(\cdot)^{-1}$ refers to vertices of $F$, whereas $\pi(\cdot)$ refers to vertices of $G^*$.
    See \cref{figure:p_star} for an illustration.

    \looseness=-1
    We now show that $\pi^{-1}(\chi_{G^*}(v))$ is a vertex separator of $F$, that is, if these vertices
    are deleted from $F$, the resulting graph is disconnected.
    The argument to derive this at its core works by observing that $F$ can be embedded into $G^*$ in a particular way (as witnessed by $\pi$), and since $G^*$ has certain structural features, these carry over to $F$, leading to a contradiction.
    
    \begin{figure}
        \begin{center}
            \includegraphics[]{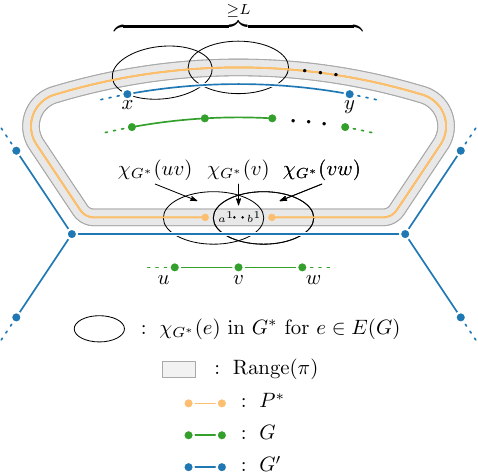}
        \end{center}
        \caption[Illustration accompanying \cref{lemma:splitting_config_of_greater_width}.]{Illustration accompanying \cref{lemma:splitting_config_of_greater_width}. The black ovals denote the edge gadgets in $G^*$, graph $G$ is displayed in green, and the underlying graph $G'$ is rendered in blue. The gray area shows the range of a hypothetical embedding $\pi$ of $F$ in $G^*$, that has to ``go around'' in the construction, since it cannot span across the intersection of edge gadgets $\chi_{G^*}(v)$.
        Additionally, in orange, the path $P^*$ traversing the embedding is shown.}
        \label{figure:p_star}
    \end{figure}

    Suppose that $\pi^{-1}(\chi_{G^*}(v))$ is not a vertex separator of $F$. Then, all neighbors of $\pi^{-1}(\chi_{G^*}(v))$ in $V(F) \setminus \pi^{-1}(\chi_{G^*}(v))$
    are pairwise connected via some path in $F$ not using any of $\pi^{-1}(\chi_{G^*}(v))$ each.
    Select any one of these paths and call it $P$.
    Without loss of generality, $P$ starts with a vertex of $(\chi_{G^*}(uv) \setminus \chi_{G^*}(v))$ and ends in a vertex of $(\chi_{G^*}(vw) \setminus \chi_{G^*}(v))$.
    Due to the existence of $\pi$, we know that $P^* \coloneqq \pi(P)$ 
    gives an isomorphic path in $G^*$.
    Since $P$ does not use vertices of $\pi^{-1}(\chi_{G^*}(v))$, $P^*$ does not use vertices of $\chi_{G^*}(v)$.

    \looseness=-1
    By construction of $G^*$, all paths connecting the first and last vertex of $P^*$ in $G^*$ that are constrained to the union of the vertex sets of both edge gadgets, that is to $\chi_{G^*}(uv) \cup \chi_{G^*}(vw)$, must 
    traverse the intersection of both edge gadgets, that is, $\chi_{G^*}(uv) \cap \chi_{G^*}(vw) = \chi_{G^*}(v)$.
    But $P^*$ does not intersect with $\chi_{G^*}(v)$, hence it cannot be one of these paths.
    Therefore, $P^*$ must traverse $G^*$ using edge gadgets the ``other way around'', that is, not use the direct connection.

    Observe that $P^*$ induces a path corresponding to the edge gadgets it traverses in $G$, which in turn induces a path of length at least three in the underlying cubic graph $G'$.
    At least one of these edges in $G'$, call it $xy$, must be fully traversed by $P^*$ in the corresponding part of~$G^*$.
    Thus, there are then $x', y' \in V(P^*)$ where $x' \in \chi_{G^*}(x) \cap V(P^*)$ and $y' \in \chi_{G^*}(y) \cap V(P^*)$.
    The distance between $x'$ and $y'$ in $G^*$ is at least $L = 2 \cdot \max_{F \in \mathcal{F}} \diam(F)$, the number of times $xy$ is subdivided in $G$.
    But then $\pi^{-1}(x')$ and $\pi^{-1}(y')$, vertices of $F$, have a distance of at least $L$ in $F$ as well, a contradiction to the choice of $L$.
    Thus, $\pi^{-1}(\chi_{G^*}(v))$ is a vertex separator of $F$.
    Furthermore, since $|\chi_{G^*}(v)| \le 2$ and $F$ is biconnected, indeed $|\chi_{G^*}(v)| = 2$;
    we shall denote the two corresponding elements by $a^1$ and $b^1$.
    
    Let $D$ be any connected component of $F \setminus \set{a^1, b^1}$.
    Suppose there is only one edge of the form $dv$ with $d \in V(D)$ and $v \in \set{a^1, b^1}$ in $E(F)$.
    Then, $F$ could not be biconnected, for the removal of a single vertex (either $a^1$ or $b^1$) would suffice to render $F$ disconnected.
    Thus, there is a path $P^1$ from $a^1$ to $b^1$ in $F$ with $P^1 \subseteq V(D) \cup \set{a^1, b^1}$. Since $G^*$ was constructed with respect to the splitting configuration $C^0$, we notice that $|P^1| \ge \width(C^0)$.

    We will carry on with exploiting the structure of $F$ to obtain a splitting configuration satisfying the conditions of this lemma. 
    Let $X$ be the vertex set of some distinct connected component of $F \setminus \set{a^1, b^1}$,
    and let $Y \coloneqq V(F) \setminus (\set{a^1, b^1} \cup X)$. We notice that $a^1b^1 \not\in E(F)$, since $\pi(a^1)$ and $\pi(b^1)$ are descendants of the same split in the construction of $G^*$.
    Furthermore, $X$ and $Y$ form a partition of $V(F) \setminus \set{a^1, b^1}$.
    Thus, we may define a new disjoint splitting configuration $C^1 \coloneqq (F, a^1, A_1^1, A_2^1, b^1, B_1^1, B_2^1)$ with
    $A_1^1 \coloneqq N_F(a^1) \cap X$,
    $A_2^1 \coloneqq N_F(a^1) \cap Y$,
    $B_1^1 \coloneqq N_F(b^1) \cap X$, and
    $B_2^1 \coloneqq N_F(b^1) \cap Y$.

    Remember that, by \cref{definition:width_splitting_configuration}, $\width(C^1) = \min \set{d_{F_1}(a_1^1, a_2^1), d_{F_2}(b_1^1, b_2^1)}$,
    where
    $F_1 \coloneqq \Split(F, a^1, A_1^1, A_2^1, a_1^1, a_2^1)$ and
    $F_2 \coloneqq \Split(F, b^1, B_1^1, B_2^1, b_1^1, b_2^1)$,
    such that $a_1^1, a_2^1, b_1^1$, and $b_2^1$ are fresh vertices.
    Consider $F_1$:
    By the argument above, we deduce that there is a shortest path through the
    descendant vertices of $X$
    from $a_1^1$ to $b^1$ in $F_1$.
    Furthermore, since $F \setminus \set{a^1, b^1}$ is comprised of at least two connected components,
    there also exists a shortest path through one of them (using descendant vertices of $Y$) from $b^1$ to $a_2^1$.
    Each of the considered shortest paths must have length at least $\width(C^0)$, as $G^*$ was constructed with respect to the splitting configuration $C^0$.
    Also, note that all paths connecting $a_1^1$ and $a_2^1$ in $F$ must traverse $b^1$.
    Thus, combining these paths yields that $d_{F_1}(a_1^1, a_2^1) \ge 2 \width(C^0)$.
    \iflong{}See \cref{figure:new_splitting_conf} for an illustration.\fi

    \iflong
        \begin{figure}
            \begin{center}
                \includegraphics[]{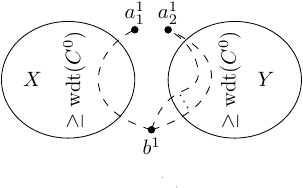}
            \end{center}
            \caption[Illustration accompanying \cref{lemma:splitting_config_of_greater_width}.]{Illustration accompanying \cref{lemma:splitting_config_of_greater_width}, displaying the increased width of the derived splitting configuration.}
            \label{figure:new_splitting_conf}
        \end{figure}
    \fi

    We proceed symmetrically for $F_2$.
    Hence, in total, we obtain that $\width(C^1) > \width(C^0)$, and conclude that $C^1$ is a splitting configuration satisfying the required conditions.
\end{proof}

With this, we have essentially arrived at an algorithm to find separating disjoint splitting configurations.
The proof proceeds inductively: Given any disjoint splitting configuration that is not separating, we will be able to leverage \cref{lemma:splitting_config_of_greater_width} above to find a new disjoint splitting configuration of increased width.
Since this width cannot grow without bound, at some point, a suitable configuration will be found.
\begin{restatable}[\appsymb]{lemma}{lemmaexistencesplittingconfigforbiconnected}\label{lemma:existence_splitting_config_for_biconnected}
    Let $\mathcal{F}$ be a family of biconnected graphs
    of bounded circumference and bounded diameter.
    Then, there exists a separating
    disjoint splitting configuration for $\mathcal{F}$.
\end{restatable}
\appendixproof{lemma:existence_splitting_config_for_biconnected}{
\ifshort\lemmaexistencesplittingconfigforbiconnected*\fi
\begin{proof}
    Select a disjoint splitting configuration $C^0 = (H^0, a^0, A_1^0, A_2^0, b^0, B_1^0, B_2^0)$
    with $H_0 \in \mathcal{F}$, which is possible since biconnected graphs have order, as well as minimum-degree, of at least two.
    We will now consider a sequence of disjoint splitting configurations that starts with $C^0$,
    which we will construct as follows: Let $C^i$ be the last member of the sequence so far.
    If $C^i$ is separating for $\mathcal{F}$, the sequence is complete. Otherwise, we obtain
    $C^{i+1}$ by applying \cref{lemma:splitting_config_of_greater_width} to $C^i$.

    Towards a contradiction, assume this sequence is infinite.
    Let $C\!\uparrow$ be the largest $k$ such that there is $F \in \mathcal{F}$ with 
    $\emb_\subseteq(C_k, F) \neq \varnothing$.
    Observe that $\emb_\subseteq(C_{\width(C^i)}, H^i) \neq \varnothing$ for all $C^i$.
    Thus, $\width(C^i) \le C\!\uparrow$ for all $C^i$.
    But $\width(C^i)$ grows without bound for increasing $i$, a contradiction.
    Therefore, the sequence is finite and its last element is a separating disjoint splitting configuration for $\mathcal{F}$.
\end{proof}
}

With a means of finding separating splitting configurations, we are left to ensure they are also intra-edge embedding-free and therefore admissible.
To that end, we impose the restriction on $\mathcal{F}$ that when any $F \in \mathcal{F}$ is destroyed by one or two disjoint splits, the resulting graph is free of forbidden (induced) subgraphs.
For example, each finite set of cycles satisfies this condition.
Finally, it remains to apply \cref{lemma:np_complete_admissible} to obtain \NP-hardness for the vertex splitting problem in question.
\begin{theorem}\label{theorem:admissible_biconnected_np_complete}
    Let $\mathcal{F}$ be a family of biconnected graphs
    of bounded circumference and bounded diameter
    such that for every $F \in \mathcal{F}$ it holds that, if 
    $F'$ is obtained from $F$ by performing at least one and at most two
    non-trivial disjoint splits, then $F' \in \free_\subseteq(\mathcal{F})$.
    Then, \textsc{$\free_\subseteq(\mathcal{F})$-VS} and \textsc{$\free_\prec(\mathcal{F})$-VS} are \NP-complete.
\end{theorem}
\begin{proof} 
    We apply \cref{lemma:existence_splitting_config_for_biconnected} to obtain a separating disjoint splitting configuration $C$ 
    for $\mathcal{F}$.
    The requirement posed on $\mathcal{F}$ allows us to deduce that $\AllConstr(K_2, C) \subseteq \free_\subseteq(\mathcal{F})$, hence $C$ is intra-edge embedding-free for $\mathcal{F}$.
    Thus, $C$ is admissible for $\mathcal{F}$ and by \cref{lemma:np_complete_admissible},
    we conclude that both \textsc{$\free_\subseteq(\mathcal{F})$-VS}, as well as \textsc{$\free_\prec(\mathcal{F})$-VS}, are \NP-hard.
\end{proof}

If there is just a single biconnected forbidden (induced) subgraph, the situation simplifies:

\begin{theorem}\label{theorem:single_biconnected_np_complete}
    Let $F$ be a biconnected graph. Then, \textsc{$\free_\subseteq(\set{F})$-VS}
    and \textsc{$\free_\prec(\set{F})$-VS} are \NP-complete.
\end{theorem}
\begin{proof}
  \looseness=-1
    We first show that for each $F'$ obtained from $F$ by performing
    at least one, but at most two non-trivial disjoint splits, it holds
    that $F' \in \free_\subseteq{\set{F}}$.
    Suppose there is $\pi \in \emb_\subseteq(F, F')$. %
    Then, because $\N{F'} > \N{F}$, we may select $v \in V(F') \setminus \range(\pi)$ such that $F \subseteq F' - v$. %
    Thus $\M{F} \leq \M{F' - v}$.
    On the other hand, by the non-triviality of the splits in question and the biconnectedness of $F$, vertex~$v$ cannot be isolated.
    Thus, $\M{F' - v} < \M{F}$, and together both inequalities imply $\M{F} < \M{F}$. This contradicts $\M{F} = \M{F}$, which is implied by the splits being disjoint.
    Therefore, $\pi$ cannot exist and $F$ possesses the claimed property.
    With this, we can invoke \cref{theorem:admissible_biconnected_np_complete} and conclude that both problems are \NP-hard.
    The \NP-membership is a trivial consequence of the fact that $\set{F}$ is finite.
\end{proof}

 \subsection{Forbidden Subgraphs of Higher Connectedness}
\label{section:higher_connected_fisgs}

As we progress onward from biconnected graphs to higher degrees of connectedness,
the restrictions imposed on the forbidden subgraphs relax.
In the case of triconnectedness, we are able to drop all restrictions on $\mathcal{F}$, except the bounded diameter:
\begin{restatable}[\appsymb]{theorem}{theoremtriconnectedboundedcircumnpcomplete}\label{theorem:triconnected_bounded_circum_np_complete}
    Let $\mathcal{F}$ be a family of triconnected graphs of bounded diameter.
    Then, \textsc{$\free_\subseteq(\mathcal{F})$-VS} and \textsc{$\free_\prec(\mathcal{F})$-VS} are \NP-hard.
\end{restatable}
\appendixproof{theorem:triconnected_bounded_circum_np_complete}{
\ifshort\theoremtriconnectedboundedcircumnpcomplete*\fi
\begin{proof}[Proof sketch]
    Let $C$ be a disjoint splitting configuration based on $\mathcal{F}$ such that each split produces a vertex of degree one and $H \in \argmin_{F \in \mathcal{F}} \M{F}$.
    We show that $C$ is admissible for $\mathcal{F}$. By \cref{lemma:np_complete_admissible} our proof is then complete.

    First, we need to establish that $C$ is separating for $\mathcal{F}$.
    We proceed in a similar vein as in \cref{lemma:splitting_config_of_greater_width}.
    Assume $C$ is not separating for $\mathcal{F}$.
    Then, there is an instantiation of our construction $G^*$ 
    such that some $F \in \mathcal{F}$ has an embedding reaching from one edge gadget to a neighboring one.
    We observe that for the same reason as in the biconnected case, with the construction being ``too elongated'', the embedding cannot take the indirect route through $G^*$. 
    But the direct route through the edge gadget ``overlap'' can support only at most two vertex-disjoint paths, contradicting that there have to be at least three by the triconnectedness of $F$.

    Secondly, we need to show that inside a single edge gadget (after it was split), no embeddings of any $F \in \mathcal{F}$ can occur.
    Since the splits produce a vertex of degree one, and a triconnected graph cannot have such vertices, the largest triconnected component remaining has fewer edges then $H$.
    By choice of $H \in \mathcal{F}$, no $F \in \mathcal{F}$ can hence be embedded into the resulting graph.
\end{proof}
}

Finally, we consider the case of 4-connected $\mathcal{F}$.
Notice that all technical restrictions on $\mathcal{F}$ have vanished.

\begin{restatable}[\appsymb]{theorem}{theoremfourconnectednpcomplete}\label{theorem:four_connected_np_complete}
    Let $\mathcal{F}$ be a family of $4$-connected graphs. 
    Then, \textsc{$\free_\subseteq(\mathcal{F})$-VS} 
    and \textsc{$\free_\prec(\mathcal{F})$-VS} are \NP-hard.
\end{restatable}
\appendixproof{theorem:four_connected_np_complete}{
\ifshort\theoremfourconnectednpcomplete*\fi
\begin{figure}
    \begin{center}
        \includegraphics{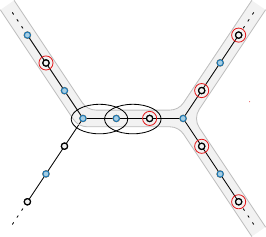}
    \end{center}
    \caption[Illustration for the proof sketch of \cref{theorem:four_connected_np_complete}.]{Illustration for the proof sketch of \cref{theorem:four_connected_np_complete}.
    The figure shows part of a 2-subdivision of a cubic graph and a vertex cover thereof (marked in blue).
    Two edge gadgets are marked with ellipses. The shaded gray area denotes a hypothetical embedding of a forbidden graph
    that ``goes around'' the construction, connecting the marked edge gadgets.
    In red, ``cutpoint-vertices'' are marked, stemming from vertices not part of the vertex cover.
    In this example, it suffices to delete two ``cutpoints'' in order to disconnect the forbidden graph; in general, this could require up to three.
    In either case, a contraction is reached, as the forbidden graph is 4-connected and can thus not be disconnected by less than four vertex deletions.}
    \label{figure:four_connected}
\end{figure}
\begin{proof}[Proof sketch]
We proceed similarly as in the triconnected case and use the same kind of splitting configuration.
Notice that we need a slightly different notion of admissibility compared to \cref{definition:admissible}:
The quantity $L$ as defined there is not meaningful in our case, as we have no bound on the diameter of the forbidden subgraphs.
Instead, we set $L \coloneqq 2$, and also use this value in the reduction itself when constructing~$G^*$.

We can apply the same argument as to why there are no embeddings of $F \in \mathcal{F}$ in edge gadgets after splitting.
The difference lies in proving the separability; the previous argument is not applicable, as there is no bound on the diameter in the given case.
Instead, we observe the following:
Suppose there is an embedding of some $F \in \mathcal{F}$ in $G^*$ that reaches from one edge gadget to a neighboring one.
The embedding cannot only use the ``direct overlap'', as then deleting at most two vertices would render the embedding disconnected, contradicting the 4-connectedness.
Thus, the embedding needs to ``go the other way'' around the construction.
But then, it fully traverses the edge gadgets that stem from either two or three edges of the 2-subdivision of a $K_{1,3}$ in the underlying cubic graph.
Notice that we can always modify the given vertex cover of the 2-subdivided cubic graph such that one vertex on each subdivided edge is not part of the vertex cover.
These vertices outside of the vertex cover were not split in the construction of $G^*$.
Thus, all paths that directly connect two edge gadgets whose intersection is a vertex not part of the vertex cover must traverse through a single ``cutpoint-vertex''.
Hence, by deleting at most three such vertices, we can render our 4-connected embedding disconnected, a contradiction.
Reference \cref{figure:four_connected} for an illustration of the argument.
\end{proof}
}

\section{Infinite Families of Forbidden Cycles}
\label{section:cycles}
\appendixsection{section:cycles}

In this section, we show that \textsc{Bipartite Vertex Splitting} and \textsc{Perfect Vertex Splitting} are \NP-complete by reducing from
\textsc{$2$-Subdivided Cubic Vertex Cover}.
In the reduction, we replace each edge of the \textsc{$2$-Subdivided Cubic Vertex Cover} instance with a triangle.
Intuitively, this forces each splitting sequence making the graph bipartite 
to ``hit'' each triangle, analogous to how a vertex cover needs to ``hit'' each edge of a graph as well.
The more difficult direction is to show how to use a vertex cover to split the constructed graph as to make it bipartite.
Here, we first apply a set of preprocessing rules to a given vertex cover. In essence, we remove vertices that are ``unnecessarily'' included in the vertex cover.
Then, we partition the constructed graph into smaller components and recognize that each component can be split according to the preprocessed vertex cover and a finite set of rules to make it 2-colorable.
We then show that all ``local'' 2-colorings compose into a ``global'' 2-coloring, hence yielding a splitting sequence that renders the constructed graph bipartite.

\begin{restatable}[\appsymb]{lemma}{lemmabipartitereduction}\label{lemma:bipartite_reduction}
    Let $(G, k)$ be an instance of \textsc{$2$-Subdivided Cubic Vertex Cover},
    and let the graph $G^*$ be obtained from $G$ by adding a new vertex $w$ and the edges $\set{wu, wv}$ to $G$ for each $uv \in E(G)$.
    Then, $(G, k)$ is a positive instance of \textsc{$2$-Subdivided Cubic Vertex Cover}
    if and only if $(G^*, k)$ is a positive instance of \textsc{Bipartite Vertex Splitting}.
\end{restatable}

\appendixproof{lemma:bipartite_reduction}{
\ifshort\lemmabipartitereduction*\fi

\begin{figure}
    \begin{center}
        \includegraphics{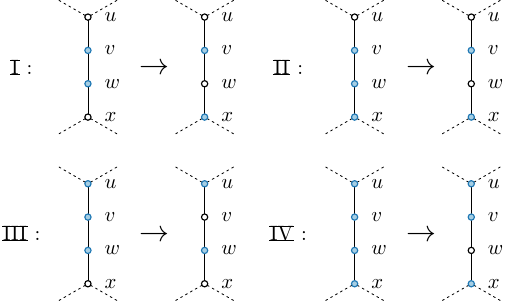}
    \end{center}
    \caption[The four preprocessing rules used in the proof of \cref{lemma:bipartite_reduction}.]{The four preprocessing rules used in the proof of \cref{lemma:bipartite_reduction}.
    Informally, for each rule, the left-hand side denotes a portion of a graph together with the corresponding part of a vertex cover, and the right-hand side shows how the vertex cover shall be altered.}
    \label{figure:preprocessing}
\end{figure}

\begin{figure}
    \begin{center}
        \includegraphics{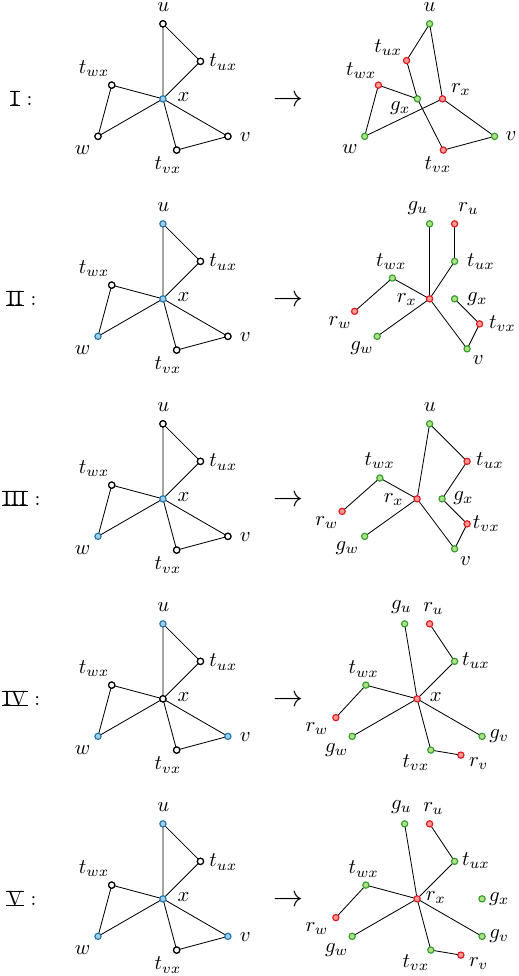}
    \end{center}
    \caption[The five rules used in the proof of \cref{lemma:bipartite_reduction}.]{The five rules used in the proof of \cref{lemma:bipartite_reduction}.
    Informally, for each rule, the left-hand side denotes a portion of a graph over the vertex set $V(G_0)$ together with the corresponding part of a vertex cover.
    If the rule matches, the right-hand side defines a graph over the vertex set $V^*$ and a two-coloring thereof.}
    \label{figure:bipartite_table}
\end{figure}

\begin{figure}
    \begin{center}
        \includegraphics{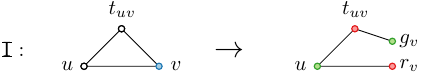}
    \end{center}
    \caption[The final rule used in \cref{lemma:bipartite_reduction}.]{The final rule used in \cref{lemma:bipartite_reduction}, analogous to the rules given in \cref{figure:bipartite_table}.  }
    \label{figure:linking_triangle}
\end{figure}

\begin{figure}
    \begin{center}
        \includegraphics{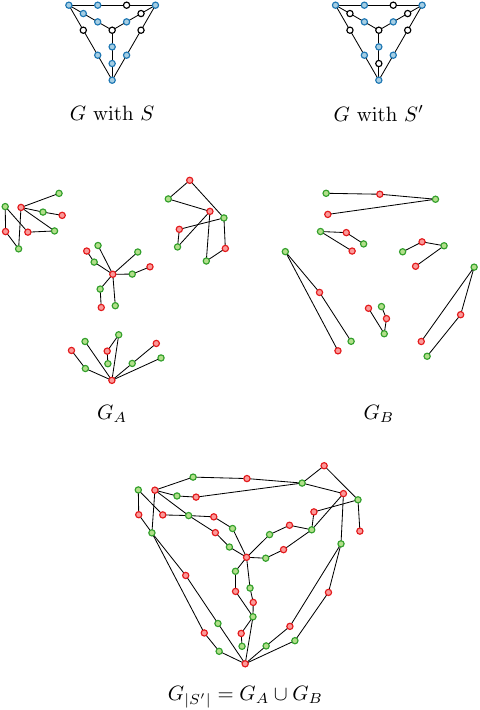}
    \end{center}
    \caption[Example for the construction used in the proof of \cref{lemma:bipartite_reduction}.]{Example for the construction used in the proof of \cref{lemma:bipartite_reduction}.}
    \label{figure:bipartite_example}
\end{figure}
    
\begin{proof}
    $(\Leftarrow)\colon$
    Let $G_0, \ldots, G_\ell$ be a splitting sequence where the first graph equals $G^*$,
    the last graph is bipartite, and $\ell \leq k$.
    We observe that $G_0 = \Constr(\vec{G}, C, \varnothing)$ (defined in \cref{section:a_construction})
    where $\vec{G}$ is some orientation of $G$,
    and $C$ is a splitting configuration based upon $\set{K_3}$.
    The class of bipartite graphs is described by $\free_\subseteq(\mathcal{F})$,
    where $\mathcal{F}$ is the set of all odd circles, implying $C$ is also based upon $\mathcal{F}$.
    Therefore, we can apply \cref{lemma:splitting_sequence_to_vertex_cover} and obtain a vertex cover $S$ of $G$ with $|S| \leq k$.

    $(\Rightarrow)\colon$
    Let $S$ be a vertex cover of $G$ using at most $k$ vertices.

    Roughly speaking, we will construct a bipartite graph by composing 2-colorable components
    selected according to the vertex cover $S$. The resulting graph will be the last graph of a $k$-splitting sequence starting with $G^*$;
    the precise order of splits will be unimportant.

    \smallskip
    \textit{Preprocessing.}
    We begin with applying preprocessing rules to $S$ to obtain a derived vertex cover $S'$ with $|S'| \leq |S|$.
    Let $S_0 \coloneqq S$. We will build a sequence $S_0, S_1, \dots$ as follows:
    Consider some path $uvwx$ in $G$ such that $d_G(u) = d_G(x) = 3$.
    Consult each left-hand side of the four rules depicted in \cref{figure:preprocessing}.
    If all of the depicted vertices marked in blue are contained in $S_i$, while all remaining depicted vertices
    are not contained in $S_i$, then we say the corresponding rule is applicable.
    If this is the case, we apply the rule by setting $S_{i + 1}$ to $S_i$, but updating the membership status of $u, v, w, x$ in $S_{i+1}$ such that
    the vertices marked blue on the corresponding right-hand side of the rule are contained in $S_{i+1}$, whereas the vertices that are not marked blue are not contained in $S_{i+1}$.
    We notice that $S_{i+1}$ is still a vertex cover, as all edges of the path $uvwx$ are covered, as depicted in \cref{figure:preprocessing},
    and all other edges of $G$ are covered by some part of $S_i$ not modified by applying the rule.
    Furthermore, observe that $|S_{i+1}|$ cannot exceed $|S_i|$.

    Repeat the above procedure until no rule is applicable anymore. Consult \cref{figure:preprocessing} and observe that this is the case eventually.
    We use $S'$ to denote the last element of the sequence $S_0, S_1, \ldots$ of vertex covers.

    \smallskip
    \textit{Construction of $G_{|S'|}$.}
    In the remainder of this proof, we will make use of three sets of new vertices:
    \begin{align*}
        V^*_\text{red}   &\coloneqq \set{r_v \mid v \in V(G) \cap S'},\\
        V^*_\text{green} &\coloneqq \set{g_v \mid v \in V(G) \cap S'}\text{, and}\\
        V^*_\text{triangle} &\coloneqq \set{t_{e} \mid e \in E(G)}.
    \end{align*}
    Later, $V^*_\text{red}$ will denote vertices produced by vertex splitting to be colored red,
    $V^*_\text{green}$ will serve a similar role for vertices to be colored green,
    and $V^*_\text{triangle}$ will serve as endpoints for triangles that we are about to introduce.
    Our constructions will be based on the vertex set 
    \begin{align*}
        V^* \coloneqq V(G) \cup V^*_\text{red} \cup V^*_\text{green} \cup V^*_\text{triangle}.
    \end{align*}

    To start, we construct the graph $G_0$ by augmenting $G$ with the vertices
    $V^*_\text{triangle}$ and the edges $\bigcup_{vw \in E(G)} \set{v v_{vw}, w v_{vw}}$,
    that is, we extend each edge of $G$ to a triangle. Note that this graph is isomorphic to $G^*$ by construction.
    Next, we define two sets of connected components of $G_0$:
    \begin{align*}
        C_A &\coloneqq \set{ G_0[\set{x, u, t_{ux}, v, t_{vx}, w, t_{wx}}] \mid x \in V(G), d_G(x) = 3, N_G(x) = \set{u, v, w}}\text{,}\\
        C_B &\coloneqq \set{ G_0[\set{u, v, t_{uv}}] \mid uv \in E(G), d_G(u) = d_G(v) = 2}.
    \end{align*}
    In other words, the set of graphs $C_A$ contains the subgraphs induced by the claws of $G$ where each edge is extended to a triangle,
    whereas $C_B$ is composed of the triangles stemming from the edges in $G$ connecting its claws, where each edge is extended to a triangle using one vertex of $V^*_\text{triangle}$ each.
    The (non-vertex-disjoint) union of all these graphs forms $G_0$, that is, $\bigcup_{C \in C_A \cup C_B} C = G_0$.

    Next, we will define a map $f_A$ that assigns to each $H \in C_A$ a graph over $V^*$.
    Alongside, we will define a coloring $c_A \colon \bigcup_{H \in C_A} V(f_A(H)) \to \set{\text{red}, \text{green}}$.
    Let $H \in C_A$.
    We use \cref{figure:bipartite_table} to obtain the value of $f_A(H)$
    and $c_A(y)$ for $y \in V(H)$ as follows:
    In the figure, five rules are depicted; consider one such rule and 
    let $H'$ be the graph depicted on the left-hand side.
    We say, the rule \emph{matches}, if there is $\pi \in \emb_\prec(H, H')$, such that for each $v \in V(H)$
    it holds that $v \in S'$ if and only if $\pi(v)$ is marked in blue in the depiction of $H'$.
    If the rule matches, $\pi$ in combination with the graph on the right-hand side specify the graph $f_A(H)$.
    Additionally, we use the depicted $\set{\text{red}, \text{green}}$-coloring to assign 
    the values of $c_A(v)$ for $v \in V(H)$.

    We observe that exactly one rule will match for each $H \in C_A$,
    as $(V(H) \cap S') \setminus V^*_\text{triangle}$ is a vertex cover of $H[V(H) \setminus V^*_\text{triangle}] \simeq K_{1,3}$,
    and the graphs depicted on the left-hand sides of the rules (when disregarding vertices of $V^*_\text{triangle}$) specify all possible vertex covers of $K_{1, 3}$ up to isomorphism.
    Furthermore, we notice that if $a \neq b \in C_A$, then $f_A(a)$ and $f_B(b)$ are vertex-disjoint.
    Thus, $c_A$ is well-defined.

    We proceed symmetrically with the definition of $f_B$ that assigns to each $H \in C_B$ a graph over $V^*$ and
    the coloring $c_B \colon \bigcup_{H \in C_B} V(f_B(H)) \to \set{\text{red}, \text{green}}$.
    Different from the last case, we use \cref{figure:linking_triangle} as the table of rules.
    We also use another justification as to why always exactly one pattern matches:
    Due to the preprocessing step where we derived $S'$ from $S$, exactly one vertex of each $H \in C_B$ that is not a member of $V^*_{\text{triangle}}$ is contained in $S'$.
    Next, we define
    \begin{align*}
        G_A \coloneqq \dot{\bigcup}_{H \in C_A} f_A(H) \qquad \text{and} \qquad G_B \coloneqq \dot{\bigcup}_{H \in C_B} f_B(H).
    \end{align*}    
    Finally, let $G_{|S'|}$ be the (non-vertex-disjoint) union of $G_A$ and $G_B$. (See \cref{figure:bipartite_example} for an example.)
    
    \smallskip
    \textit{Properties of $G_{|S'|}$.}
    To understand the rationale for constructing $G_{|S'|}$ this way, we make a series of observations.
    Notice that for all $A_1 \neq A_2 \in C_A$, the graphs $f_A(A_1)$ and $f_B(A_2)$ are vertex-disjoint.
    Additionally, an analog statement holds for $C_B$.
    Thus, $c_A$ describes a proper 2-coloring of $G_A$,
    whereas $c_B$ describes a proper 2-coloring of $G_B$.
    Next, remember that all graphs in $G_A$ and $G_B$ were defined using the common vertex set $V^*$.
    Further, notice that the only way for two graphs $g(H_1)$ and $h(H_2)$ with $\set{g, h} = \set{f_A, f_B}$ and $H_1 \neq H_2 \in C_A \cup C_B$ to share common vertices is
    when, without loss of generality, $g = f_A, h = f_B, H_1 \in C_A$ and $H_2 \in C_B$.
    In this case, we have that $V(f_A(H_1)) \cap V(f_B(H_2)) = \set{r_v, g_v} \subseteq V^*_\text{red} \cup V^*_\text{green}$
    for some $v \in V(G)$ if $v \in S'$; otherwise, $v \not\in S'$ and $V(f_A(H_1)) \cap V(f_B(H_2)) = \set{v}$.
    Using this, we make the crucial observation that the definitions of $f_A$ and $f_B$ were chosen precisely such that when the 2-sets 
    $\set{ \set{r_v, g_v} \mid v \in S'}$
    are merged in $G_{|S'|}$ respectively, a graph isomorphic to $G_0$ is obtained.

    We claim that $c_A \cup c_B$ provides a proper 2-coloring of $G_{|S'|}$.
    To this end, we show that if $x \in \dom(c_A) \cap \dom(c_B)$, then $c_A(x) = c_B(x)$.
    Let $x \in \dom(c_A) \cap \dom(c_B)$.
    First, assume that $x \subseteq V^*_{\text{red}} \cup V^*_{\text{green}}$.
    Then, $c_A(x) = c_B(x)$, as both in \cref{figure:bipartite_table} as well as in \cref{figure:linking_triangle}
    vertices of $V^*_{\text{red}}$ are always assigned the color red, whereas vertices of $V^*_{\text{green}}$ are always assigned the color green.
    Otherwise, $x \in V(G)$. Then, again inspecting \cref{figure:bipartite_table} and \cref{figure:linking_triangle} yields that $c_A(x) = c_B(x) = \text{green}$.
    Therefore, we conclude that $G_{|S'|}$ admits a proper 2-coloring, or phrased differently: $G_{|S'|}$ is bipartite.

    Finally, select a total order $s_1, \ldots, s_{|S'|}$ of $S'$.
    We will construct a $|S'|$-splitting sequence in reverse, starting with its last element, $G_{|S'|}$.
    To obtain $G_{i-1}$ from $G_i$ for $i \in \set{1, \ldots, |S'|}$, consider vertex $s_i$.
    We observe that $r_{s_i}$ and  $g_{s_i}$ are two independent vertices of $V(G_i)$.
    Now, $G_{i-i}$ is obtained from $G_i$ by merging $r_{s_i}$ and $g_{s_i}$ into $s_i$.
    Thus, the initial graph of the sequence obtained this way equals $G_0$ as defined above.
    Therefore, $G_0, \dots, G_{|S'|}$ certifies that $(G^*, k)$ is a positive instance of \textsc{Bipartite Vertex Splitting}.
\end{proof}
}

Using this reduction the hardness proof is immediate:
\begin{restatable}[\appsymb]{theorem}{theorembipartitevsnpcomplete}\label{theorem:bipartite_vs_np_complete}
    \textsc{Bipartite Vertex Splitting} is \NP-complete.
\end{restatable}
\appendixproof{theorem:bipartite_vs_np_complete}{
\ifshort\theorembipartitevsnpcomplete*\fi
\begin{proof}
    The \NP-membership of \textsc{Bipartite Vertex Splitting} is trivial.
    To derive its \NP-hardness, we use the equivalence given in \cref{lemma:bipartite_reduction}
    to obtain a polynomial-time many-one reduction from the 
    \NP-hard \textsc{$2$-Subdivided Cubic Vertex Cover} (\cref{lemma:two_l_subdivided_cubic_vc_np_hard}) problem to \textsc{Bipartite Vertex Splitting}.
\end{proof}
}

\subsection{Splitting to Perfect Graphs}

Finally, we extend our argument to show that \textsc{Perfect Vertex Splitting} is \NP-complete too.
By the Strong Perfect Graph Theorem \cite{perfect_graph_theorem}, a graph is perfect if and only if it contains no induced cycle of odd length greater than three (``odd holes'') and no induced complement of such a cycle (``odd antiholes'').
Exploiting this, we base our construction upon $C_5$ instead of $K_3$ and provide an appropriate table of rules to obtain a bipartite graph, analogous to the previous case.
Finally, we use a simple sparsity argument to show that when splitting according to these rules, no sufficiently large odd antiholes are introduced.

\begin{restatable}[\appsymb]{theorem}{theoremperfectvsnpcomplete}\label{theorem:perfect_vs_np_complete}
    \textsc{Perfect Vertex Splitting} is \NP-complete.
\end{restatable}

\appendixproof{theorem:perfect_vs_np_complete}{
\ifshort\theoremperfectvsnpcomplete*\fi

\begin{figure}
    \begin{subfigure}{\linewidth}
        \begin{center}
            \includegraphics{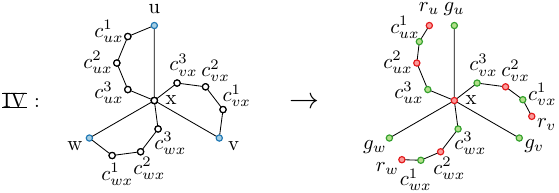}\hfill    
        \end{center}
        \caption[Analogue of \cref{figure:bipartite_table} for the case of \textsc{Perfect-Vertex Splitting}.]{Analogue of \cref{figure:bipartite_table}.
        Of the five rules, only rule \barroman{IV} is shown, which is the analog of \barroman{IV} in \cref{figure:bipartite_table}.
        The other four rules can be obtained by taking the correspondingly numbered rule from \cref{figure:bipartite_table}, and for each ``triangle-vertex'', $t_e$, subdividing both incident edges one time.
        The coloring on the right-hand side is then obtained by carrying over the colors of all descendants of $u, v, w, x$. This fixes the colors of the remaining vertices as well.
        }
        \label{figure:perfect_table}
    \end{subfigure}\par\medskip
    \begin{subfigure}{\linewidth}
        \begin{center}
            \includegraphics{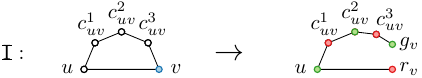}\hfill    
        \end{center}
        \caption[Analogue of \cref{figure:linking_triangle} for the case of \textsc{Perfect-Vertex Splitting}.]{Analogue of \cref{figure:linking_triangle}.}
        \label{figure:linking_c5}
    \end{subfigure}\par\medskip
    \caption[The rules used in the proof of \cref{theorem:perfect_vs_np_complete}.]{The rules used in the proof of \cref{theorem:perfect_vs_np_complete}. Their semantics are completely analogous to the rules used to show the \NP-hardness of \textsc{Bipartite Vertex Splitting}, but instead cater to \textsc{Perfect Vertex Splitting}. }
\end{figure}

\begin{proof}
    We use the same technique as we did when deriving the \NP-completeness of \textsc{Bipartite Vertex Splitting}.
    Thus, we only describe how to alter \cref{lemma:bipartite_reduction} and \cref{theorem:perfect_vs_np_complete}.

    The instance mapping changes as follows:
    Given an instance $(G, k)$ of \textsc{$2$-Subdivided Cubic Vertex Cover},
    instead of extending each edge $uv$ of $G$ to a triangle, we add the path $u c^1_{uv} c^2_{uv} c^3_{uv} v$ for each such edge, where $c^1_{uv}, c^2_{uv}, c^3_{uv}$ are new vertices.
    This has the effect of replacing each edge with a $C_5$, instead of a $C_3$ as was the case in the bipartite case.
    Instead of using $V^*_\text{triangle}$, we collect these new vertices in $V^*_{C_5}$, that is,
    \begin{align*}
        V^*_{C_5} &\coloneqq \set{c^i_{e} \mid e \in E(G), i \in \set{1, 2, 3}}.
    \end{align*}
    Furthermore, everywhere in the proof where a vertex $t_e$ is mentioned, replace this occurrence with $c^1_e, c^2_e, c^3_e$.

    $(\Leftarrow)\colon$ In the backward direction of the correctness proof, we proceed completely analogously, but select a splitting configuration based upon $\set{C_5}$ instead of $\set{K_3}$.

    $(\Rightarrow)\colon$ In the forward direction of the correctness proof, we need to make the following alterations:
    Instead of using the rules defined in \cref{figure:bipartite_table} and \cref{figure:linking_triangle}, we use \cref{figure:perfect_table} and \cref{figure:linking_c5} respectively.
    Then, using the same argument as in the \textsc{Bipartite Vertex Splitting} case, $G_{|S'|}$ is bipartite and hence does not contain induced odd cycles.
    It remains to show that $G_{|S'|}$ is also free of induced antiholes of odd length, that is, $\set{\overline{C_{2i+1}} \mid i \in \mathbb{N}^+}$.
    The first such graph, $\overline{C_5}$, is isomorphic to $C_5$ and hence not contained in $G_{|S'|}$.
    Observe that the minimum vertex degree in any of $\overline{C_7}, \overline{C_9}, \ldots$ is at least four.
    Now, suppose that one of $\overline{C_7}, \overline{C_9}, \ldots$ is an induced subgraph of $G_{|S'|}$.
    Then, any such embedding cannot use vertices of $G_{|S'|}$ with degree two.
    But the induced subgraph of $G_{|S'|}$ where all vertices with degree two were deleted clearly has a maximum vertex degree of three.
    Therefore, such an embedding cannot exist and $G_{|S'|}$ is free of all odd antiholes, and thereby also perfect.
\end{proof}
}

\newcommand{\kTriangleFreeVS}[1]{\textsc{Triangle-Free $#1$-Vertex Splitting}}
\newcommand{\TriangleFreeVS}[0]{\textsc{Triangle-Free Vertex Splitting}}
\newcommand{\TriangleFreeCol}{\textsc{Triangle-Free 3-Colorability}}

\newcommand{\ShallowTriangleVS}{\textsc{Shallow Triangle-Free Vertex Splitting}}
\newcommand{\kShallowTriangleVS}[1]{\textsc{Shallow Triangle-Free $#1$-Vertex Splitting}}

\section{Parameterized Complexity}
\label{section:parameterized}
\appendixsection{section:parameterized}

\looseness=-1
In this section, we derive that \TriangleFreeVS{} is para-\NP-hard, but show that when we may only split each vertex at most once, the problem becomes \XP-tractable.
(For the relevant notions from parameterized algorithmics we refer to the literature~\cite{fpt_downey_fellows,fptbook,FlumG06,Niedermeier06}.)

For a fixed $k$, we call the restriction of \TriangleFreeVS{} to instances $(G, \ell)$ where $\ell = k$, \kTriangleFreeVS{k}.
We will show that \kTriangleFreeVS{k} is \NP-hard for $k \ge 2$, and hence,
\TriangleFreeVS{} parameterized by the number of splits is para-\NP-hard.
To this end, we reduce from the \NP-hard \TriangleFreeCol{} problem \cite{complexityColoringHFree} to \kTriangleFreeVS{k}.
In essence, we take an instance $G$ of the \TriangleFreeCol{} and add a universal vertex $u$.
Then, the triangles in the resulting graph correspond precisely to all ``coloring constraints'', as $G$ is triangle-free.
Furthermore, each triangle contains $u$.
Now, suppose that $u$ is split exactly two times using only disjoint splits, such that all triangles get destroyed.
Then, no edge $xy$ of $G$ can have two adjacent edges with any one of the three descendants of $u$.
Thus, there is an edge from $x$ to one of the three descendants, and an edge from $y$ to a different descendant.
As there are precisely three descendants, we can use this structure to obtain a proper 3-coloring of $G$.

It is left to ensure that $u$ must be split as described above in the resulting instance.
For this, we perform two measures:
Firstly, we add $k - 2$ disjoint triangles that are not connected to the universal vertex.
This fixes all but two splits.
Secondly, we do not just use $G$, but rather three disjoint copies of $G$.
This way, it is ensured that at least one copy of $G$ is only split at the vertex $u$.
From this copy, we will be able to extract the desired 3-coloring.

\begin{restatable}[\appsymb]{theorem}{theoremparaNPhard}\label{theorem:paraNPhard}
    \kTriangleFreeVS{k} is \NP-hard for $k \ge 2$.
\end{restatable}
\appendixproof{theorem:paraNPhard}{
\ifshort\theoremparaNPhard*\fi
\begin{figure}
    \begin{subfigure}{\linewidth}
        \begin{center}
            \includegraphics{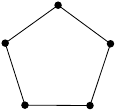}\hfill    
        \end{center}
        \caption{A triangle-free graph $G$.}
        \label{figure:c5}
    \end{subfigure}\par\medskip
    \begin{subfigure}{\linewidth}
        \begin{center}
            \includegraphics{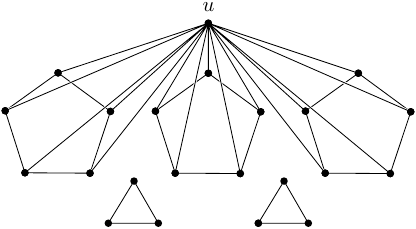}\hfill    
        \end{center}
        \caption{The reduction applied to $G$, resulting in $G^*$ and the instance $(G^*, 4)$ of \kTriangleFreeVS{4}.
        }
        \label{figure:3col_reduction}
    \end{subfigure}\par\medskip
    \begin{subfigure}{\linewidth}
        \begin{center}
            \includegraphics{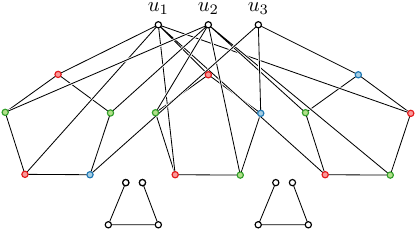}\hfill    
        \end{center}
        \caption{The final graph of a splitting sequence of length four starting with $G^*$ and ending with a triangle-free graph. The neighborhoods of $\set{u_1, u_2, u_3}$ induce three proper 3-colorings of $G$.}
        \label{figure:3col_reduction_split}
    \end{subfigure}\par\medskip
    \caption[Example for the reduction used in the proof of \cref{theorem:paraNPhard}.]{Example for the reduction used in the proof of \cref{theorem:paraNPhard} with $k = 4$. }
    \label{figure:3col_reduction_example}
\end{figure}
\begin{proof}
    Let $k \ge 2$.
    We perform a polynomial-time reduction from the \NP-hard \TriangleFreeCol{} problem, which decides whether a triangle-free graph admits a proper 3-coloring \cite{complexityColoringHFree}.
    Let an instance of this problem be given by a graph $G$.
    To build an instance $(G^*, k)$, of \kTriangleFreeVS{k}, we take the disjoint union of three copies of $G$, call them $G^1, G^2, G^3$, and add a new universal vertex $u$.
    Finally, we add $k - 2$ disjoint copies of $K_3$.
    See \cref{figure:3col_reduction_example} for an example.

    We claim that $G$ is a positive instance of \TriangleFreeCol{} if and only if $(G^*, k)$ is a positive instance of \kTriangleFreeVS{k}.

    $(\Rightarrow)\colon$
    Assume $G$ is 3-colorable.
    Let $c \colon V(G) \to \set{1, 2, 3}$ be a proper 3-coloring of $G^1 \cupdot G^2 \cupdot G^3$, where each copy of $G$ is colored according to some proper 3-coloring of $G$.
    We construct the splitting sequence $G_0^*, \dots, G_k^*$ with $G_0^* = G^*$ as follows:
    The first two splits are used to split 
    the universal vertex $u$ two times using disjoint splits; call its descendants $u_1, u_2$ and $u_3$.
    First, we split $u$ and assign each edge $vu \in E(G_0^*)$ to $u_1$ if $c(v) = 1$ and to $u_2$ otherwise.
    In the second split, we split $u_2$ and assign each edge $vu_2 \in E(G_1^*)$ to $u_2$ if $c(v) = 2$ and to $u_3$ otherwise.
    We use the remaining $k - 2$ spits to destroy the remaining $k - 2$ disjoint copies of $K_3$.
    
    We claim that $G_2^*$ is triangle-free apart from the $k - 2$ disjoint copies of $K_3$.
    Suppose not.
    Then, there is a triangle that does not use any vertices of the $k - 2$ disjoint copies of $K_3$.
    Observe that in the first two splits, only the universal vertex $u$ (and descendants of $u$) were split.
    Also, $G^1, G^2, G^3$ are triangle-free and there are no edges between them in $G_2^*$.
    Thus, the triangle contains a descendant of $u$, call it $u_i$, and two additional vertices that both stem from a shared copy of $G$; call them $x, y$.
    But then, by the construction of the splitting sequence, we find that $i = c(x) = c(y)$.
    But $c(x) \neq c(y)$, since $c$ is a proper 3-coloring of $G^1 \cupdot G^2 \cupdot G^3$.
    Thus $G_2^*$ is indeed triangle-free apart from the $k - 2$ disjoint copies of $K_3$.
    As the last $k - 2$ splits that destroy one of the $k - 2$ disjoint triangles each do not introduce new triangles, we conclude that $G_k^*$ is triangle-free.

    $(\Leftarrow)\colon$
    Let $G_0^*, \dots, G_\ell^*$ be a splitting sequence with $G_0^* = G^*$, $\ell \leq k$, and $G_\ell^* \in \free_\subseteq(K_3)$.
    We assume that the sequence consists of exactly $k$ splits, since,
    in case the given sequence is shorter, we can pad the sequence to the desired length using trivial splits that only create isolated vertices.
    Furthermore, we can assume that the sequence only performs disjoint splits. In case a split was not disjoint, we can simply remove the duplicated edges throughout the sequence.
    Clearly, removing edges cannot introduce triangles to a triangle-free graph.
    Finally, we may assume that precisely the last $k - 2$ splits were used to destroy the $k - 2$ disjoint copies of $K_3$ in $G^*$.
    Thus, already $G_2^*$ must be free of triangles apart from the $k - 2$ disjoint copies of $K_3$.

    In the case that already $G_0^*$ is triangle-free apart from the $k - 2$ disjoint copies of $K_3$, $G$ is edgeless and thus 3-colorable.
    Otherwise, $G_0^*$ contains a triangle not using any of the vertices of the $k - 2$ disjoint copies of $K_3$.
    Suppose that $u$ was not split in the sequence,
    then at least one $G$-copy, together with all of the triangles it forms with $u$, must still be intact in $G_2^*$, as there are three copies, but only two splits in the sequence that are not used to destroy the $k - 2$ disjoint copies of $K_3$.
    As $G_2^*$ is triangle-free apart from the $k - 2$ disjoint copies of $K_3$, this cannot be. Hence, $u$ was split in the splitting sequence.

    Next, we select a suitable copy $G^i$ with $i \in \set{1, 2, 3}$ of $G$ which we will use to construct a 3-coloring of $G$.
    In case $u$ was split two times, we let $i \coloneqq 1$.
    Otherwise, a vertex of a copy $G^j$ with $j \in \set{1, 2, 3}$ was split. Then, let $i \in \set{1, 2, 3} \setminus \set{j}$.
    Regardless of the applicable case, none of $V(G^i)$ were split in the splitting sequence.
    We use $u_1, u_2, \ldots$ to denote the descendants of $u$ in $G_2^*$.
    Then, we define a coloring $c \colon V(G^i) \to \set{1, 2, 3}$ of $G^i$ as follows:
    Consider each $v \in V(G^i)$. The edge $vu$ was assigned to exactly one of $u_1, u_2, \ldots$ in the splitting sequence.
    Let the index be denoted by $j$ and set $c(v) \coloneqq j$.

    We claim that $c$ is a proper 3-coloring of $G^i$.
    Towards a contradiction, suppose there is $xy \in E(G^i)$ with $c(x) = c(y)$.
    Then, by the definition of $c$, there is $j \in \set{1, 2, 3}$, such that $xu_j, yu_j$ are edges in $G_2^*$, implying that the vertices $\set{x, y, u_j}$ form a triangle in $G_2^*$.
    But as $G_2^*$ is triangle-free apart from the $k-2$ disjoint copies of~$K_3$, we derived a contradiction and conclude that $c$ is a proper 3-coloring of $G^i \simeq G$.
\end{proof}
}

The ``root'' of the para-\NP-hardness we have just observed seems to stem from the ability to split a single vertex more than once.
Indeed, as we will see shortly, if we remove said ability, the resulting parameterized problem is a member of \XP.
We denote the restriction of \TriangleFreeVS{}, where each vertex can only be split at most once as \ShallowTriangleVS{}.
Note that \ShallowTriangleVS{} is still an \NP-hard problem:
Since $K_3$ is biconnected, we can apply the proof of \cref{theorem:single_biconnected_np_complete} and deduce that \ShallowTriangleVS{} is \NP-complete,
as the argument remains valid even if each vertex can be split at most once.

The outline of the algorithm is as follows:
Consider an instance $(G, k)$.
We formulate a Boolean formula, $\psi$, that has a model such that at most $k$ variables of a certain kind are allowed to be true if and only if $(G, k)$ is a positive instance of \kShallowTriangleVS{k}.
Intuitively, an interpretation $I$ of $\psi$ specifies how $G$ should be split.
We make this notion precise and describe a mapping that assigns each such interpretation $I$ a graph $G^I$.
Using these definitions, we prove that the encoding $\psi$ of the instance $(G, k)$ is correct.
Finally, we show that the satisfiability of the formula (where at most $k$ variables of a certain kind are allowed to be true) can be decided efficiently by guessing part of $\psi$'s variables and solving the reduced formula in linear time using an algorithm for \textsc{2-SAT}.
In total, we obtain a running time of $\mathcal{O}( \sqrt{2}^{k^2} \cdot \N{G}^{k+3} )$.

\newcommand{\var}[1]{\protect\overrightarrow{\vphantom{abcxy}#1} }

\begin{restatable}[\appsymb]{theorem}{theoremShallowTriangleVSinXP}\label{theorem:ShallowTriangleVSinXP}
    \ShallowTriangleVS{} parameterized by $k$, the number of splits, is in \XP.
\end{restatable}
\appendixproof{theorem:ShallowTriangleVSinXP}{
\ifshort\theoremShallowTriangleVSinXP*\fi
\begin{proof}
    Let $(G, k)$ with $G = (V, E)$ be an instance of \ShallowTriangleVS{}.
    We will derive an algorithm that decides the instance in time $ \mathcal{O}( \sqrt{2}^{k^2} \cdot \N{G}^{k+3} )$.

    \smallskip
    \textit{Definition of $\psi$.}
    Let ``$<$'' be some total order on $V$.
    We define the propositional formula $\psi$ over the set of variables
    \begin{equation*}
        X \coloneqq V \cup \bigcup_{uv \in E \;\land\; u < v}\set{ \var{uv}, \var{vu} }.
    \end{equation*}
    In other words, $X$ contains a variable for each vertex of $G$, and two variables for each edge.
    The subset $V$ of $X$ will model which vertices to split, while the remaining variables will model how to select the neighborhoods of the two descendant vertices when performing a split.
    An \emph{interpretation} of $\psi$ is a set $I \subseteq X$. If a variable $x \in X$ is in $I$, we say that $x$ is true under $I$ and write $I \models x$.
    Otherwise, we say $x$ is false under $I$ and write $I \not\models x$.
    We also extend this notion to sentences over $X$ using the standard semantics of Boolean logic.
    The symbol ``$\oplus$'' denotes the exclusive or operator, ``$\supset$'' is the implication operator, and ``$\equiv$'' is used to compare literals for equality.

    Let $\mathcal{T}$ denote the set of triangles in $G$, where each triangle is given by a three-tuple of its vertices ordered according to ``$<$''.
    For each such triangle $T = (a, b, c) \in \mathcal{T}$, we define the formula
    \begin{equation*}
        \varphi(T) \coloneqq \parens{a \land \parens{\var{ba} \oplus \var{ca}}} \lor \parens{b \land \parens{\var{ab} \oplus \var{cb}}} \lor \parens{c \land \parens{\var{ac} \oplus \var{bc}}}.
    \end{equation*}
    Then, we set
    \begin{equation*}
        \psi \coloneqq \bigwedge_{T \in \mathcal{T}} \varphi(T).
    \end{equation*}

    Intuitively, for a given $T = (a, b, c) \in \mathcal{T}$, the formula $\varphi(T)$ encodes the condition for $T$ to be destroyed via vertex splitting.
    More precisely, at least one of $a, b$, or $c$ needs to be split such that the edges it forms with the other two vertices after the split do not intersect.
    Given this, the formula $\psi$ simply encodes that said condition should hold for all triangles of $G$.

    \smallskip
    \textit{Mapping $G$ to $G^I$.}
    Before we prove that $\psi$ actually describes the problem at hand properly, we first need some additional notation.
    Let $I \subseteq X$, that is, an interpretation of $\psi$.
    We construct the graph $G^I$ over the vertex set $ V \setminus I $, augmented with fresh vertices $v_1, v_2$ for each $v \in I \cap V$.

    \begin{figure}[t]
        \begin{center}
            \includegraphics{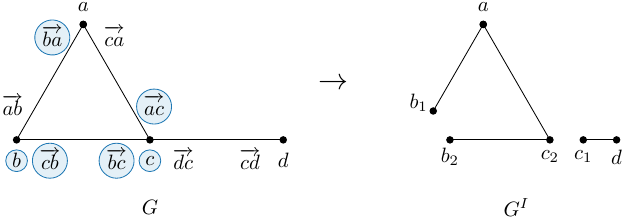}
        \end{center}
        \caption[Example for the mapping defined in \cref{theorem:ShallowTriangleVSinXP}.]{Example for the mapping defined in \cref{theorem:ShallowTriangleVSinXP}.
        A graph $G$ is mapped to $G^I$, where $I = \set{b, c, \var{ba}, \var{ac}, \var{cb}, \var{bc}} \subseteq X$ is a model of $\psi = \parens{a \land \parens{\var{ba} \oplus \var{ca}}} \lor \parens{b \land \parens{\var{ab} \oplus \var{cb}}} \lor \parens{c \land \parens{\var{ac} \oplus \var{bc}}}$.
        The variables that are true under $I$ are marked in the figure.
        Note that the same triangle-free $G^I$ would result if $\var{ba}$ were not set to true in $I$.
        }
        \label{figure:interpretation_to_graph}
    \end{figure}

    The edges of $G^I$ are obtained by mapping each edge $uv$ of $G$ to $uv^I = u'v'$, where
    \begin{equation*}
        u' \coloneqq 
        \begin{cases}
            u & \text{if } u \not\in I\\
            u_1 & \text{if } u \in I \land \var{vu} \not\in I\\
            u_2 & \text{if } u \in I \land \var{vu} \in I
        \end{cases}
        \quad \text{and} \quad
        v' \coloneqq 
        \begin{cases}
            v & \text{if } v \not\in I\\
            v_1 & \text{if } v \in I \land \var{uv} \not\in I\\
            v_2 & \text{if } v \in I \land \var{uv} \in I
        \end{cases}.
    \end{equation*}
    
    Intuitively, $G^I$ describes the graph $G$ split according to $I$.
    The subset $V \cap I$ of the interpretation indicates which vertices of $G$ should be split.
    Each vertex $v \in M \cap V$ is replaced by two descendants, $v_1$ and $v_2$.
    The split neighborhoods of $v_1$ and $v_2$ are determined by all variables in $X$ of the shape $\var{uv}$ where $u \in V$.
    If $\var{uv} \not\in I$, then the edge $uv$ (or a descendant of the edge if $u \in V$ as well) is assigned to $v_1$.
    If otherwise $\var{uv} \in I$, then the edge is assigned to $v_2$.
    Reference \cref{figure:interpretation_to_graph} for an example.

    It is easy to see that there is a bijection between the set of all possible splitting sequences of $G$ that only use disjoint splits and do not split any vertex more than once,
    and the set of possible interpretations of $X$.
    Furthermore, note that the order of splits does not matter.
    
    \smallskip
    \textit{Correctness of $\psi$.}
    We claim that $(G, k)$ is a positive instance of \kShallowTriangleVS{k} if and only if $\psi$ has a model $M \subseteq X$ with $|M \cap V| \leq k$.

    $(\Rightarrow)\colon$
    Let $G_0, \dots, G_\ell$ be a splitting sequence that splits each vertex at most once and obeys $\ell \leq k$ as well as $G_\ell \in \free_\prec(K_3)$.
    We assume that the sequence only performs disjoint splits; a sequence that uses non-disjoint splits can be transformed into such a sequence by removing the duplicated edges.
    Also, removing edges cannot introduce triangles. Hence, the last graph of the transformed sequence would be triangle-free as well.

    Let $I \subseteq X$ such that $G^I = G_\ell$.
    We claim that $I \models \psi$.
    Suppose the contrary.
    Then, there is $(a, b, c) = T \in \mathcal{T}$ such that $M \not\models \varphi(T)$.
    In case $\set{a, b, c} \cap M = \varnothing$, no vertex of $T$ was split, hence $G_\ell$ is not triangle-free, a contradiction.
    Otherwise, we observe that for $M \not\models \varphi(T)$, it needs to be the case that
    \begin{align*}
        M &\models a \supset ( \var{ba} \equiv \var{ca} ), \\
        M &\models b \supset ( \var{ab} \equiv \var{cb} ) \text{, and}\\
        M &\models c \supset ( \var{ac} \equiv \var{bc} ).
    \end{align*}
    But then, if one of $a, b, c$ is split in the splitting sequence, the edges of the triangle $T$ are always assigned to the same descendant in each split.
    Thus, $G_\ell$ contains a triangle, contradicting that $G_\ell$ is triangle-free.

    $(\Leftarrow)\colon$
    Let $M \subseteq X$ be a model of $\psi$ with $\ell \coloneqq |M \cap V| \leq k$.
    Fix a sequence of sets $S_0, \ldots, S_\ell$ with $\varnothing = S_0 \subset S_1 \subset \dots \subset S_{\ell-1} \subset S_\ell = M \cap V$.
    Then, $G_0, \dots, G_\ell$ with $G_i \coloneqq G^{S_i \cup (M \setminus V)}$ is a splitting sequence with $G_0 = G$ that splits each vertex at most once, has length $\ell \leq k$, and uses only disjoint splits.
    It remains to show that $G_\ell = G^M$ is triangle-free.
    Towards a contradiction, suppose there are $a, b, c \in V(G_\ell)$ such that $G_\ell[\set{a, b, c}] \simeq K_3$.
    The reverse operation to a disjoint vertex split is to merge two non-adjacent vertices, such that the number of edges remains invariant.
    Merging non-adjacent vertices clearly cannot destroy triangles.
    Thus, we can trace the triangle $a, b, c$ found in $G_\ell$ backward through the sequence,
    and find ancestors $a^*, b^*, c^* \in V(G_0)$ of $a, b, c$ respectively that form a triangle $T^*$ in $G_0$.
    Then, we find that $M \models \varphi(T^*)$ by choice of $M$.
    Hence, at least one of $a^*, b^*, c^*$ is in $M \cap V$.
    Without loss of generality, assume that $a^* \in M \cap V$.
    Therefore, $M$ assigns $\var{b^*a^*}$ and $\var{c^*a^*}$ different truth values.
    This implies that ${b^*a^*}^M = ba \in E(G_\ell)$ and ${c^*a^*}^M = ca \in E(G_\ell) $ are not adjacent in $G_\ell$.
    Hence, $a, b, c$ do not form a triangle in $G_\ell$, contrary to what we assumed.

\smallskip
\textit{Deciding $\psi$.}
Finally, we devise a procedure to check whether $\psi$ admits a model $M \subseteq X$ where at most $k$ variables of $V \subseteq X$ are set to true.
We search for a model by guessing the truth values of a constrained part of the variables of $\psi$.
Then, with this partial assignment in place, the remaining formula becomes equivalent to an instance of $\textsc{2-SAT}$, which is a problem belonging to the class \P.
For any $M$ that models $\psi$, the set $S \coloneqq M \cap V$ must be a hitting set of $\mathcal{T}$ of size at most $k$.
Clearly, there are at most $\mathcal{O}(\N{G}^k)$ candidate hitting sets $S$ to consider.
Based on the choice of $S$, we guess a partial truth assignment:
The variables $V \cap S$ are set to true, and the variables $V \setminus S$ are set to false.
Additionally, we identify the subset $X' \subseteq X \setminus V$ of all variables that correspond to edges in $G[S]$, that is,
\begin{equation*}
    X' \coloneqq \bigcup_{u, v \;\in\; S \colon uv \;\in\; E} \set{ \var{uv}, \var{vu} }.
\end{equation*}
We guess the truth value of all $x' \in X'$.
In total, at most $2^{\binom{|S|}{2}} = \mathcal{O}(\sqrt{2}^{k^2})$ partial truth assignments of $X$
for each choice of $S$ are enumerated.
Consider one such partial truth assignment.
For each conjunct $\varphi( (a, b, c) )$ of $\psi$, we replace the guessed variables by either $\top$ or $\bot$.
After simplifying the resulting formula, we obtain an equivalent set of clauses, in which each clause contains no more than two literals, in all cases:
\begin{description}
    \item[Case $|\set{a, b, c} \cap S| = 0\colon$]
    Impossible, as $S$ is a hitting set of $\mathcal{T}$.

    \item[Case $|\set{a, b, c} \cap S| = 1\colon$]
    Without loss of generality, assume $\set{a, b, c} \cap S = \set{a}$.
    Then, $\varphi( (a, b, c) )$ simplifies to
    $\var{ba} \oplus \var{ca}$,
    which is equivalent to the clause set
    $\set{\var{ba} \lor \var{ca}, \neg\var{ba} \lor \neg\var{ca}}$.

    \item[Case $|\set{a, b, c} \cap S| = 2\colon$]
    Without loss of generality, assume $\set{a, b, c} \cap S = \set{a, b}$.
    As $\set{a, b} \subseteq S$, the truth values of $\var{ab}$ and $\var{ba}$ are already fixed.
    Then, $\varphi( (a, b, c) )$ reduces to a clause set $\set{\var{ca^*} \lor \var{cb^*}}$, where $\var{ca^*}$ and $\var{cb^*}$ are literals of $\var{ca}$ and $\var{cb}$ respectively.

    \item[Case $|\set{a, b, c} \cap S| = 3\colon$]
    In this case, the truth values of all variables occurring in $\varphi( (a, b, c) )$ are already fixed.
    Thus, $\varphi( (a, b, c) )$ reduces to either $\set{\top}$ or $\set{\bot}$.
\end{description}
In total, this yields an instance of \textsc{2-SAT}.
There are at most $\binom{\N{G}}{3}$ triangles in $G$.
For each triangle, we obtain at most two clauses, in which at most four unique variables occur.
Thus, the size of the \textsc{2-SAT} instance, that is, the sum of the number of clauses and the number of variables, is in $\mathcal{O}(\N{G}^3)$.
By the classical result of Even, Itai, and Shamir, \textsc{2-SAT} admits a linear time algorithm \cite{linear2SAT}.

Considering all steps, that is, enumerating all hitting sets $S$ of size at most $k$, guessing the truth values of $X'$ for each such $S$, and finally deciding the resulting instance of \textsc{2-SAT}, gives a combined running time of 
$ \mathcal{O}( \sqrt{2}^{k^2} \cdot \N{G}^{k+3} )$.
\end{proof}
}

\section{Conclusion}
We have mostly obtained \NP-hardness results for \pPVSlong~(\pPVS).
However, because of the nontrivial polynomial-time algorithms for sets of small forbidden subgraphs, such as $\{\overline{K_3}, P_3\}$, and other so-far sporadic tractability results, such as for \pPVSa{Forest}, the line of separation between tractability and intractability is much more jagged than for instance for the vertex-deletion operation.
Our results show that well-connected forbidden subgraphs are an important driver of hardness, but the case of more fragile forbidden subgraphs is relatively open apart from hardness for \pFVS{P_3}.
One way to drive this direction forward is to settle the complexity of \pFVS{P_4} and \pFVS{K_{1, 3}}.
For the former, similar to \pFVS{P_3} and sigma clique covers, we can show a relation to a cograph-covering problem which we tend to believe is \NP-hard.
The latter we consider fully~open.

\bibliography{refs}

\ifshort
\cleardoublepage
\appendix
\section*{Appendix}
\label{sec:appendix}
\appendixText

\else
\fi

\end{document}